\documentclass[a4paper,amsmath,amssymb,aps,twocolumn,nofootinbib]{revtex4-2}
\pdfoutput=1 

\usepackage{braket}
\usepackage{amsmath,amssymb,amsfonts}
\usepackage{amsthm}
\usepackage{mathtools}
\usepackage{graphicx}
\usepackage{xcolor}
\usepackage[colorlinks=true,
            allcolors=blue]{hyperref}
\usepackage[capitalize,nameinlink]{cleveref}
\usepackage{comment}
\usepackage{dsfont}
\usepackage{enumitem}

\newtheorem{theorem}{Theorem}
\newtheorem{definition}{Definition}
\newtheorem{lemma}{Lemma}

\theoremstyle{definition}
\newtheorem{remark}[theorem]{Remark}

\newcommand{\f}{\forall}
\newcommand{\e}{\exists}
\newcommand{\ot}{\otimes}

\newcommand{\mcl}{\mathcal}
\newcommand{\lam}{\lambda}
\newcommand{\ep}{\varepsilon}

\newcommand{\supp}{\operatorname{supp}}
\newcommand{\diag}{\operatorname{diag}}

\newcommand{\macro}{\mathfrak{m}}

\def\mE{\mathcal{E}}

\def\mP{\mathcal{P}}

\def\sH{\mathcal{H}}

\def\openone{\mathds{1}}

\newcommand{\povm}[1]{\boldsymbol{#1}}

\renewcommand{\set}[1]{\mathcal{#1}}

\newcommand{\ketbra}[1]{|#1\rangle\!\!\; \langle #1 |}
\newcommand{\tr}[1]{\operatorname{Tr}\!\left[#1\right]}

\newcommand{\norm}[1]{{\left\|#1\right\|}}

\renewcommand{\>}{\rangle}

\renewcommand{\ge}{\geqslant}
\renewcommand{\le}{\leqslant}
\renewcommand{\geq}{\geqslant}
\renewcommand{\leq}{\leqslant}

\begin{document}

\title{On the generic increase of observational entropy in isolated systems}

\author{Teruaki Nagasawa}\email{teruaki.nagasawa@nagoya-u.jp}
\affiliation{Department of Mathematical Informatics, Nagoya University, Furo-cho Chikusa-ku, Nagoya 464-8601, Japan}

\author{Kohtaro Kato}\email{kokato@i.nagoya-u.ac.jp}
\affiliation{Department of Mathematical Informatics, Nagoya University, Furo-cho Chikusa-ku, Nagoya 464-8601, Japan}

\author{Eyuri Wakakuwa}\email{e.wakakuwa@gmail.com}
\affiliation{Department of Mathematical Informatics, Nagoya University, Furo-cho Chikusa-ku, Nagoya 464-8601, Japan}

\author{Francesco Buscemi}\email{buscemi@nagoya-u.jp}
\affiliation{Department of Mathematical Informatics, Nagoya University, Furo-cho Chikusa-ku, Nagoya 464-8601, Japan}

\date{\today}

\begin{abstract}
Observational entropy -- a quantity that unifies Boltzmann's entropy, Gibbs' entropy, von Neumann's macroscopic entropy, and the diagonal entropy -- has recently been argued to play a key role in a modern formulation of statistical mechanics.
Here, relying on algebraic techniques taken from Petz's theory of statistical sufficiency and on a L\'evy-type concentration bound, we prove rigorous theorems showing how the observational entropy of a system undergoing a unitary evolution chosen at random tends to increase with overwhelming probability and to reach its maximum very quickly.
More precisely, we show that for any observation that is sufficiently coarse with respect to the size of the system, regardless of the initial state of the system (be it pure or mixed), random evolution renders its state practically indistinguishable from the uniform (i.e., maximally mixed) distribution with a probability approaching one as the size of the system grows.
The same conclusion holds not only for random evolutions sampled according to the unitarily invariant Haar distribution, but also for approximate 2-designs, which are thought to provide a more physically and computationally reasonable model of random evolutions.
\end{abstract}

\maketitle

\section{Introduction}
\label{section:introduction}
John von Neumann, in his book on the mathematical foundations of quantum theory~\cite{von1955mathematical}, immediately after having introduced and operationally motivated the quantity that is now known as \textit{von Neumann entropy}, notices however that such a quantity is not the right one to consider in the context of statistical mechanics. This is because, as he writes,
\begin{quote}
[von Neumann entropy] is invariant in the normal [i.e., Hamiltonian] evolution in time of the system, and only increases with measurements -- in the classical theory (where the measurements in general played no role) it increased as a rule even with the ordinary mechanical evolution in time of the system~\cite{von1955mathematical}. (Square brackets added for clarity.)
\end{quote}
In the above passage, von Neumann presumably refers to the free expansion of an ideal gas, in which there is a strict increase in the (macroscopic) thermodynamic entropy, although the (microscopic) von Neumann entropy associated with the state of the gas does not change as the gas undergoes Hamiltonian evolution.

To meet this challenge, von Neumann proposes the concept of \textit{macroscopic entropy}, which takes into account not only the intrinsic uncertainty associated with the microscopic state of the system, but also the additional uncertainty associated with the coarse-grained, macroscopic observation with which the system is being monitored. As the gas expands, it is this latter aspect of uncertainty, arising from the limited capabilities of a macroscopic observer, that increases while the microscopic degrees of freedom evolve undisturbed.

Since von Neumann's proposal, macroscopic entropy has been largely overshadowed by its more famous -- and eponymous -- sibling. A notable exception is Wehrl's review paper~\cite{wehrl-1978-general-entropy}, where macroscopic entropy (therein referred to as \textit{coarse-grained entropy}) plays an important role. Nevertheless, von Neumann's macroscopic entropy and a generalization of it called \textit{observational entropy} have recently been the subject of renewed interest~\cite{safranek2019a,safranek2019b,safranek2021brief,buscemi2022observational,bai2023observational}, in connection with the mathematical and conceptual foundations of statistical mechanics~\cite{strasberg-winter-2021-PRX-quantum} and various applications~\cite{riera2020finite,safranek2020classical,deutsch2020probabilistic,faiez2020typical,nation2020snapshots,strasberg2021clausius,hamazaki2022speed,modak2022observational,sreeram2023witnessing,schindler2023continuity,safranek2023work,safranek2023measuring,safranek2023ergotropic}. 

von Neumann was able to bring the idea of macroscopic entropy to fruition by proving a powerful H-theorem~\cite{vonNeumann1929translation}, showing that macroscopic entropy tends to increase, even in Hamiltonian systems, and typically grows to its maximum value regardless of the initial state of the system~\cite{Tolman-2010-dover}. Unfortunately, however, von Neumann's result was later misunderstood and criticized, and was forgotten for decades. The facts are very nicely recounted in the references~\cite{Goldstein2010Commentary,Goldstein2010NormalTypicality}, which also have the merit of bringing von Neumann's theorem back to the attention of the community.

In this paper, we aim to derive statements similar in spirit to von Neumann's H-theorem, but valid in the more general case of observational entropy, and informed by recent results on random quantum circuits. In particular, we study the change of observational entropy that can occur in isolated\footnote{Here we follow Ref.~\cite{strasberg-winter-2021-PRX-quantum} and call a system \textit{isolated} if it can only exchange work -- its evolution is therefore unitary, but not necessarily energy conserving.} systems.
Even in this simple case, very little is known about the dynamics of observational entropy: while Ref.~\cite{strasberg-winter-2021-PRX-quantum} has shown that the observational entropy of an isolated system initialized in a state fully known to the observer cannot decrease, the conditions for its \textit{strict} increase, which is the real crux of the problem, have not been discussed. Also, nothing is known about the \textit{generic} behavior of observational entropy, i.e., what happens if the initial state of the system is arbitrarily given and the unitary evolution is drawn at random.

In order to fill these gaps, we first provide an explicit characterization of all situations in which the observational entropy undergoes a \textit{strict} increase with time. Such a characterization relies on Petz's theory of statistical sufficiency~\cite{petz1986sufficient,petz1988sufficiency,petz2003monotonicity,jencova-petz-2006-sufficiency-survey}. We then move to the case of arbitrary initial states, for which, based on a L\'evy-type concentration bound~\cite{ledoux-2001,hayden-2006-aspects-generic-ent} that we prove for the observational entropy, we arrive at a statement similar to von Neumann's H-theorem: for any observation that is ``sufficiently coarse-grained'' with respect to the size of the system, under the action of \textit{Haar-random evolution}, the observational entropy approaches its maximum, i.e., the state of the system becomes practically indistinguishable from the maximally mixed (uniform) one, regardless of the state it started from.

Nevertheless, both physical and computational insights indicate that, rather than the Haar-random unitary model, a more reasonable model for random evolutions is provided by \textit{approximate 2-designs}, i.e., finite sets of unitary operators which, when sampled at random, are indistinguishable from a Haar-random evolution up to the second moment. On the one hand, from a physical point of view, approximate 2-designs turn out to be closely related to the \textit{eigenstate thermalization hypothesis}: see, e.g., Ref.~\cite{kaneko2020characterizing} and references therein. On the other hand, from a computational point of view, it is known that while the Haar distribution requires exponentially many random quantum gates, approximate 2-designs can be generated by extremely shallow (i.e., polylog-depth) random quantum circuits~\cite{schuster2024random,laracuente2024approximate}. Motivated by these arguments, we therefore consider the concentration of the observational entropy under the action of approximate 2-designs and, by specializing several derandomization techniques~\cite{dankert-2009-2-designs,low2009large,brandao2021models,harrow2023approximate}, we show that the observational entropy also reaches its maximum value in this case, regardless of the initial state of the system. In this case, since approximate 2-designs can be viewed as very short random circuits, we can intuitively say that the maximum of the observational entropy is reached ``very quickly''. In considering thermalization (e.g., ETH) and entropy increase, it is important to consider chaos and coarse-graining (i.e., discard information). In this study, chaos is expressed in terms of $\ep$-approximate t-design, and coarse-graining is quantified in terms of coarseness of measurement (Definition \ref{def:coarse}).

The paper is structured as follows. First, in Section~\ref{section:background}, we introduce the notations and basic concepts that we will use in this paper. In Section~\ref{section:OE_increase_in_macroscopic_states}, we define and explicitly characterize macroscopic states and prove that under unitary time evolution, if the system is initialized in a non-trivial macroscopic state, its observational entropy strictly increases except for a zero-measure set of unitary operators. In Section~\ref{section:OE_increase_in_arbitrary_states}, we consider the case of arbitrary initial states, and using L\'evy-type concentration inequalities, we show that the observational entropy increases generically also in this case. This is proved separately for Haar-random unitary evolutions and $\ep$-approximate 2-designs.





\section{Background}
\label{section:background}
Following the standard conventions in quantum information theory~\cite{nielsen_chuang_2010,wilde_2013}, in this paper we consider a finite $d$-dimensional quantum system, with Hilbert space $\sH$, whose states are represented by density operators $\rho\ge 0$, $\tr{\rho}=1$. The maximally mixed (i.e., uniform) state is denoted $u=\openone/d$. The von Neumann entropy, i.e., the microscopic entropy, is defined by the formula $S(\rho)=-\tr{\rho\log\rho}$, which is zero if and only if the state $\rho$ is pure, i.e., a rank-one projector on some unit vector $|\varphi\>\in\sH$. Another central quantity is the Umegaki quantum relative entropy~\cite{umegaki-q-rel-ent-1961,umegaki1962conditional}, defined as $D(\rho\|\sigma)=\tr{\rho(\log\rho-\log\sigma)}$, where $\sigma>0$ is an invertible reference (or \textit{prior}) state. Whenever $\rho$ and $\sigma$ commute, the Umegaki relative entropy coincides with the Kullback--Leibler divergence~\cite{kullback1951information}. An observation (measurement) on the system is mathematically represented by a positive operator-valued measure (POVM), i.e., a family $\povm{P}=\{P_x\}_x$ of positive semi-definite operators $P_x\ge 0$, labeled by a finite set $\set{X}=\{x\}$ (the outcome set), and normalized so that $\sum_xP_x=\openone$: given the state of the system $\rho$, the expected probability of observing outcome $x$ is computed as $p_x=\tr{P_x\;\rho}$. Whenever all the elements of a POVM are projections, i.e., $P_xP_{x'}=\delta_{xx'}P_x$, we speak of a projection-valued measure, or PVM. A PVM is \textit{trivial} if one of its elements is the identity operator $\openone$ (and all the remaining elements are null). In what follows, as very often done in the literature, it will be convenient to think of an observation as a quantum-to-classical channel, i.e., a map $\mP(\cdot)=\sum_x\tr{P_x\;\cdot}\ketbra{x}$, where $|x\>$ are orthonormal vectors in an auxiliary Hilbert space with dimension equal to the size of the outcome set $\set{X}$. Finally, we recall the idea of POVM post-processing~\cite{martens1990nonideal,buscemi-2005-clean-POVMs}: given two POVMs $\povm{P}=\{P_x\}_{x\in\set{X}}$ and $\povm{Q}=\{Q_y\}_{y\in\set{Y}}$, defined on the same Hilbert space $\sH$ but with possibly different outcome sets, we write $\povm{Q}\preceq \povm{P}$ whenever there exists a conditional probability distribution $p(y|x)$ such that $Q_y=\sum_xp(y|x)P_x$, for all $y\in\set{Y}$.

\subsection{Observational entropy}
The observational entropy (OE) of a microscopic state $\rho$ with respect to a POVM $\povm{P}=\{P_x\}_{x\in\set{X}}$ is defined as~\cite{safranek2019a,safranek2019b,safranek2021brief,strasberg-winter-2021-PRX-quantum,buscemi2022observational,bai2023observational}
\begin{align*}
    S_{\povm{P}}(\rho)&=-\sum_xp_x\log\frac{p_x}{V_x}\nonumber\\
    &=\log d-D(\mP(\rho)\|\mP(u))\;,
\end{align*}
where $p_x=\tr{P_x\;\rho}$ and $V_x=\tr{P_x}$. Since $V_x=0$ implies $P_x=0$ and hence, in particular, $p_x=0$, without loss of generality we can consider only POVMs with $V_x>0$ for all $x\in\set{X}$, so that the OE is always finite.  von Neumann's original definition of \textit{macroscopic entropy} is limited to PVMs. The presence of both the probabilities $p_x$ and the volume terms $V_x$ suggests that OE somehow ``interpolates'' between Boltzmann's and Gibbs' entropies. Indeed, when there exists one particular $\bar{x}$ such that $p_{\bar{x}}=1$, OE recovers the Boltzmann entropy $\log V_{\bar{x}}$, whereas when the volume terms are all equal to one, OE coincides with the Gibbs entropy $-\sum_xp_x\log p_x$. In the latter case, in particular, if the POVM consists of the orthogonal projectors on the energy eigenbasis, then OE recovers what is known as \textit{diagonal entropy}~\cite{polkovnikov2011microscopic}. It is straightforward from the definition that the difference of $S_{\povm{P}}(\rho)$ from its maximum value $\log{d}$ indicates how distinguishable state $\rho$ is from $u$ by measurement $\povm{P}$.

The fundamental bound of OE, which is a consequence of the data-processing property of the Umegaki relative entropy, is $S_{\povm{P}}(\rho)\ge S(\rho)$, which holds for any choice of $\rho$ and $\povm{P}$~\cite{buscemi2022observational}. States that saturate the bound, i.e., states $\rho$ such that $S_{\povm{P}}(\rho)= S(\rho)$ are called \textit{macroscopic} for $\povm{P}$. The reason for such a name comes from the fact that the condition $S_{\povm{P}}(\rho)= S(\rho)$ holds if and only if~\cite{buscemi2022observational}
\begin{align}\label{eq:macrostate}
    \rho=\sum_x\tr{P_x\;\rho}\frac{P_x}{V_x}\;,
\end{align}
which means that the state $\rho$ can be inferred~\cite{buscemi2020thermodynamic,buscemi-scarani-2021fluctuation,aw-buscemi-scarani} only from the knowledge of the measurement $\povm{P}$ and its outcomes' statistics $p_x$. This is all and the only information that is available to the macroscopic observer~\cite{buscemi2022observational}. In this precise sense, then, states that satisfy the equality $S_{\povm{P}}(\rho)= S(\rho)$ are called macroscopic. Notice that the maximally mixed state $u$ is always macroscopic, for any choice of POVM $\povm{P}$. Moreover, $u$ always achieves the maximum value of OE, that is, $S_{\povm{P}}(u)=\log d$.

\section{Observational entropy increase in macroscopic states}
\label{section:OE_increase_in_macroscopic_states}
As anticipated in the introduction, one of the main reasons to consider OE is that it can increase even in isolated systems, in contrast to von Neumann entropy, which instead remains constant. Motivated by Ref.~\cite{strasberg-winter-2021-PRX-quantum}, we begin our study by considering the behavior of OE when the initial state of the system is macroscopic. Such an assumption will be lifted in the rest of the paper.

Let us thus consider an isolated system evolving in time from $t=t_0$ to $t=t_1>t_0$. Let $\rho_0$ be the initial state of the system, $U$ describe the time evolution from $t_0$ to $t_1$, and $\rho_1=U\rho_0U^\dag$ be the state of the system at $t_1$. Let us also assume that, at time $t_0$, the system's state is macroscopic for $\povm{P}$. While the von Neumann entropy $S(\rho_t)$ of the system remains constant (as a consequence of the fact that the von Neumann entropy only depends on the spectrum of the density operator, which does not change under unitary transformations), for the OE we have:
\begin{align}
    S_{\povm{P}}(\rho_1)&=-\sum_x\tr{P_x\;\rho_1}\log\frac{\tr{P_x\;\rho_1}}{\tr{P_x}}\nonumber\\
    &=\sum_x\tr{U^\dag P_xU\;\rho_0}\log\frac{\tr{U^\dag P_xU\;\rho_0}}{\tr{U^\dag P_xU}}\nonumber\\
    &=S_{U^\dag \povm{P}U}(\rho_0)\nonumber\\
    &\ge S(\rho_0)=S_{\povm{P}}(\rho_0)=S(\rho_1)\;.\label{eq:inequality-u}
\end{align}
The final inequality holds because $\rho_0$ is macroscopic for $\povm{P}$, but may not be so for $U^\dag\povm{P}U$. Thus, from the above, we immediately see that:
\begin{enumerate}[label=\roman*)]
    \item the OE of an isolated system starting in a macroscopic state never decreases (cfr. Lemma~5 in~\cite{strasberg-winter-2021-PRX-quantum});
    \item it remains constant if and only if $\rho_1$ is \textit{also} macroscopic for the \textit{same} $\povm{P}$ as $\rho_0$.
\end{enumerate}
Given that, the question that we want to consider now is: when does the OE \textit{strictly} increase? In order to answer this question, we first need to provide a characterization of all macroscopic states associated with a given POVM $\povm{P}=\{P_x\}$. While Eq.~\eqref{eq:macrostate} provides an implicit characterization, the following theorem provides it \textit{explicitly}.
\begin{theorem}[Macroscopic states]\label{th:macro-uniform}
    Given a POVM $\povm{P}=\{P_x\}$, a state $\macro$ is macroscopic for $\povm{P}$, i.e., satisfies Eq.~\eqref{eq:macrostate}, if and only if there exists a PVM $\povm{\Pi}=\{\Pi_y\}_y$, with $\povm{\Pi}\preceq \povm{P}$, together with coefficients $c_y\ge 0$, such that
    \begin{align}\label{eq:macro-explicit}
        \macro=\sum_yc_y\Pi_y\;.
    \end{align}
\end{theorem}

The full proof of Theorem~\ref{th:macro-uniform}, which is based on the theory of statistical sufficiency~\cite{petz2003monotonicity,jencova-petz-2006-sufficiency-survey}, can be found in Appendix~\ref{appendix:characterization}; here we only comment on its consequences. We first note that, since the maximally mixed state $u$ is macroscopic for any POVM but remains invariant, the interesting situations, i.e., those in which a strict increase of OE can occur, may arise only if non-uniform macroscopic states exist. As a consequence of Theorem~\ref{th:macro-uniform}, a necessary condition for a state $\macro$ to be macroscopic for a POVM $\povm{P}=\{P_x\}_{x\in\set{X}}$, is that $[\macro,P_x]=0$ for all $x\in\set{X}$. This fact, whose proof can be found in Appendix~\ref{appendix:characterization}, immediately tells us that, in the case of an isolated system initially prepared in a non-uniform macroscopic state, only a very restricted set of unitary operators, i.e., those that satisfy the conservation-like relation
\begin{align*}
    [U\macro U^\dag,P_x]=0\;,\quad\forall x\;,
\end{align*}
can preserve the observer's information about the system, whereas a generic evolution, such as one uniformly sampled from the entire set of unitary operators, will necessarily cause a strict increase in OE. In such cases, although the microscopic evolution is perfectly reversible, from the macroscopic observer's point of view, information is irreversibly lost.

\section{Observational entropy increase in arbitrary states}
\label{section:OE_increase_in_arbitrary_states}
In the above discussion, we used algebraic arguments to treat the case of isolated systems that are initially prepared in non-uniform macroscopic states. In what follows, we instead apply measure-theoretic ideas, in particular L\'evy-type concentration bounds, to show that, under appropriate assumptions, no matter what the initial state of the system is (macroscopic or not, pure or mixed), if the observation is ``sufficiently coarse'' with respect to the size of the system, the probability that a random unitary evolution will bring the OE of the system close to its maximum value $\log d$, so that the state of the system becomes macroscopically indistinguishable from the uniform distribution, is close to one.

In order to formalize this intuition we begin with the simplest, although highly idealized, case of a random evolution sampled from the Haar (unitarily invariant) distribution. (Proof in Appendix~\ref{appendix:concentration}.)

\begin{theorem}[Haar-random case]\label{th:concentration}
    Let us consider a $d$-dimensional system in an arbitrary (but fixed) state $\rho$, a POVM $\povm{P}=\{P_x\}$ with a finite number of outcomes, and a value $\delta >0$. For a unitary operator $U$ sampled at random according to the Haar (unitarily invariant) distribution, the probability that the system's observational entropy $S_{\povm{P}}(U\rho U^\dag)$ is $\delta$-far from the maximum value $\log d$ can be bounded as follows:
    \begin{align}
        &\mathbb{P}_{H}\left\{S_{\povm{P}}(U\rho U^\dag)\le(1-\delta)\log d\right\}\nonumber\\
        &\qquad\qquad\le \frac{4}{\kappa(\povm{P})}e^{-C\delta\kappa(\povm{P})^{2}d\log{d}}\;,
        \label{eq:concentration}
    \end{align}
    where $\kappa(\povm{P})=\min_x\tr{P_x\;u}=\frac1d\min_xV_x$ and $C=\frac{1}{18\pi^3}\approx0.0018>2^{-10}$.
\end{theorem}

In other words, after a random unitary evolution, under any sufficiently coarse observation, the system will essentially look as if it were in the  maximally mixed state, no matter what its actual initial state was. Looking at Eq.~\eqref{eq:concentration}, the parameter that decides whether the POVM is sufficiently coarse with respect to given dimension $d$ and tolerance $\delta$ is the parameter $\kappa(\povm{P})\in(0,1)$, which is a \textit{universal} parameter, i.e., independent of the system's initial state.
For the right-hand side of Eq.~\eqref{eq:concentration} to be small, $\kappa(\povm{P})$, which is the ratio between the smallest volume term of the POVM and the total dimension $d$, has to play well with $C$, $\delta$, and $d$. This condition is reminiscent of von Neumann's condition on the minimum volume of the phase cells in his proof of the H-theorem~\cite{vonNeumann1929translation}:
\begin{quote}
    The number of states (quantum orbits) in each phase cell has to be not only very large, but also on average quite large compared to the number of phase cells.
\end{quote}
In our notation, the role of the ``number of states in each phase cell'', that is, the minimum thereof, is played by $\min_x V_x=d\kappa(\povm{P})$, while the ``number of phase cells'' is just the number of possible outcomes, i.e., the number of the elements of the POVM $\povm{P}$. Let us denote the latter by $N(\povm{P})$.
We can thus summarize von Neumann's condition as $N(\povm{P})\ll d\kappa(\povm{P})$.

The normalization of the POVM, i.e., $\sum_xP_x=\openone$, implies that $N(\povm{P})\kappa(\povm{P})\le\sum_x\tr{P_x\;u}= 1$ or, equivalently, $N(\povm{P})\le 1/\kappa(\povm{P})$. Hence, if $1/\kappa(\povm{P})\ll d\kappa(\povm{P})$, that is, if 
\begin{align}\label{eq:von-neu-condition}
\sqrt{d}\ll d\kappa(\povm{P})=\min_x\tr{P_x}\;,
\end{align}
then von Neumann's condition is satisfied. A schematic illustration is provided in Fig.~\ref{fig:sizes}.
As we will show next, condition~\eqref{eq:von-neu-condition} is closely related to the requirement that the right-hand side in~\eqref{eq:concentration} goes to zero as $d$ goes to infinity.

\subsection{Taking the limit}
The condition~\eqref{eq:von-neu-condition} requires that $\min_x\tr{P_x}$ is much larger than $\sqrt{d}$, while remaining smaller than $d/N(\povm{P})$ (as a consequence of the normalization). One might wonder how restrictive this requirement is. As the following back-of-the-envelope calculation shows, it is not particularly onerous. This is because, already for moderate dimensions, there is a lot of room available between $\sqrt{d}$ and $d$: see Fig.~\ref{fig:sizes} for an illustration. To get an idea of the orders of magnitude we are talking about, let us consider a system comprising 128 qubits ($d=2^{128}$) and take, for example, $\delta=2^{-5}$ and $\min_x\tr{P_x}=2^{90}$, i.e., $\kappa(\povm{P})=2^{-38}$. Using these numbers, the exponent in~\eqref{eq:concentration} is bounded as
\begin{align*}
    C\delta\kappa(\povm{P})^2d\log{d}>2^{-10}2^{-5}\frac{2^{180}}{2^{128}}2^7=2^{44}\;,
\end{align*}
meaning that the probability that the OE, after a Haar-random evolution, does \textit{not} exceed $(1-2^{-5})\times 128=124$ bits is less than $2^{40}\times\exp[-2^{44}]\approx 0$.

\begin{figure}[t]
    \centering
    \includegraphics[width=0.8\linewidth]{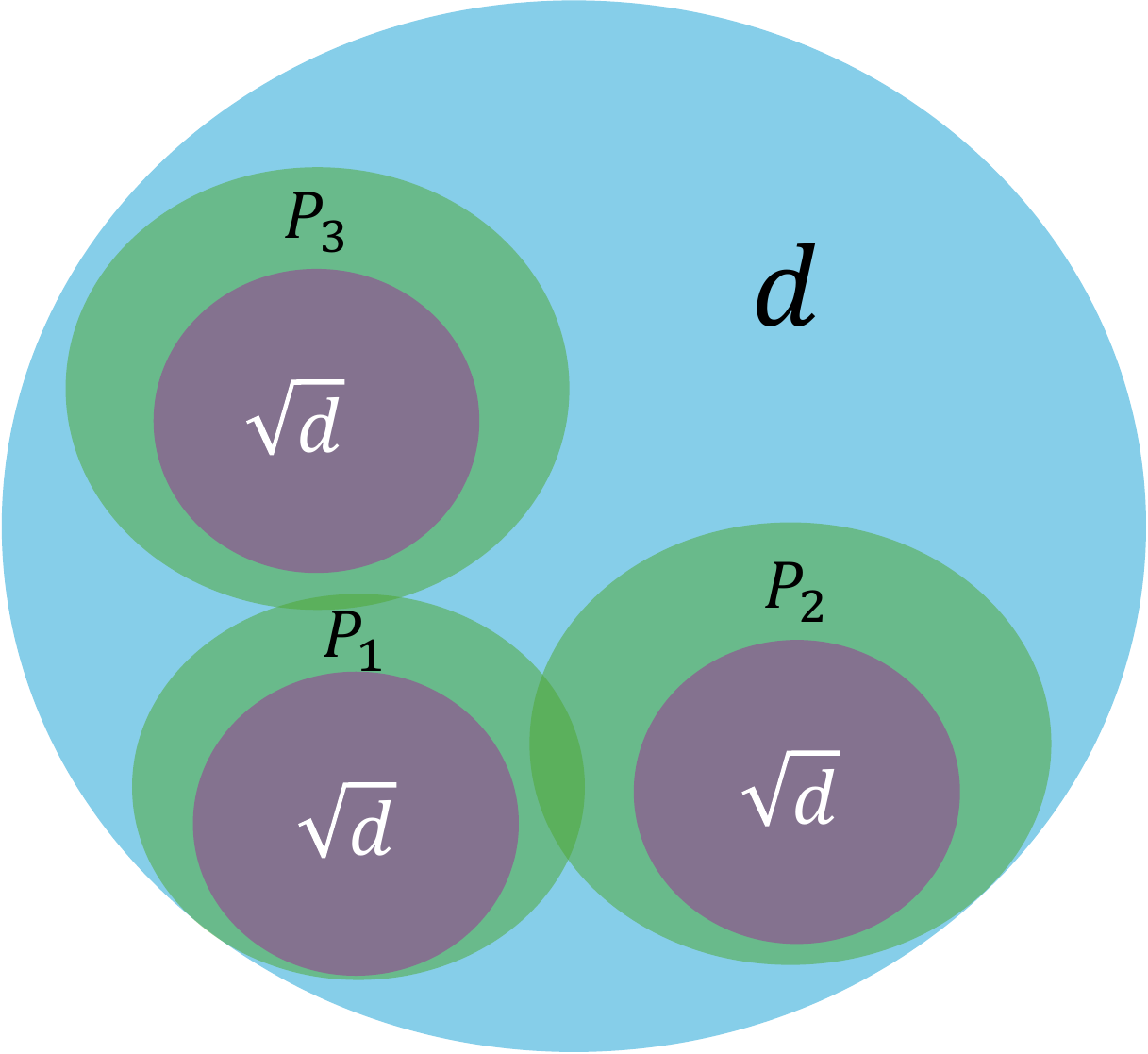}
    \caption{A schematic illustration to help visualize, in a non-rigorous but intuitive way, \textit{von Neumann's condition}, as formalized in Eq.~\eqref{eq:von-neu-condition}. The outermost light blue circle represents the total Hilbert space, which has dimension $d$. The inner green circles represent the POVM elements: the circles overlap, which means that the POVM elements need not be orthogonal. The dark gray inner core represents the von Neumann bound Eq.~\eqref{eq:von-neu-condition}: the green circles should all be ``quite larger'' than their inner cores. As $d$ grows, there is a lot of room between $d$ and $\sqrt{d}$ to accommodate many POVM elements.}
    \label{fig:sizes}
\end{figure}

While Theorem~\ref{th:concentration} above provides a bound on the probability of large deviations when all parameters (i.e., $\kappa(\povm{P})$, $d$, and $\delta$) are \textit{fixed}, a better understanding of their interplay can be obtained by looking at the asymptotic behavior of~\eqref{eq:concentration}. To this purpose, let us consider a sequence of systems and observations, in the limit of $d$ (i.e., the system's Hilbert space dimension) tending to infinity. We must first define exactly what it means for a sequence of observations to be ``asymptotically coarse''. A possible (but by no means unique) definition is as follows:
\begin{definition}[Asymptotic coarseness]\label{def:coarse}
    Consider a sequence of systems with increasing dimension $d$ and, in each system, a POVM $\povm{P}^{(d)}=\{P^{(d)}_{x_d}\}_{x_d}$. For each $d$, define $\kappa(d)\equiv\kappa(\povm{P}^{(d)})=\min_{x_d}\tr{P^{(d)}_{x_d}\;u}$. We say that the sequence of POVMs $\{\povm{P}^{(d)}\}_{d\in\mathbb{N}}$ is \emph{asymptotically coarse} whenever there exists $\tau>0$ such that
    \begin{align}
        \kappa(d)=\Omega(d^{-\frac{1}{2}+\tau})\;,
        \label{eq:kkappadorder}
    \end{align}
    i.e., whenever $\e{M}>0$ and $\e{d_{0}}$ such that
    \begin{align}
       \kappa(d)\geq M\cdot d^{-\frac{1}{2}+\tau}\;,\quad \f{d}>d_{0}\;.
    \end{align}
\end{definition}

The above definition can be justified starting from von Neumann's condition~\eqref{eq:von-neu-condition} as follows.
Let us assume that, in the limit $d\to \infty$, von Neumann's condition becomes $N(d)/d\kappa(d)\to 0$. For $\kappa(d)\sim d^{\alpha}$, and using the fact that $N(d)\kappa(d)\le1$ (consequence of the normalization of the POVM), we obtain
\begin{align*}
 \frac{N(d)}{d\;\kappa(d)}&\le   \frac1{d\;\kappa(d)^2}\\
 &\sim d^{-2\alpha-1}\;,
\end{align*}
which goes to zero if and only if $\alpha>-1/2$, in agreement with Definition~\ref{def:coarse}.

An alternative justification for Definition~\ref{def:coarse} can be derived directly from Theorem \ref{th:concentration}. For an asymptotically coarse sequence $\{\povm{P}^{(d)}\}_{d\in\mathbb{N}}$ of POVM, i.e., such that $\kappa(d)=\Omega(d^{-\frac{1}{2}+\tau})$, the right-hand side of Eq.~(\ref{eq:concentration}) is of order $d^{\frac{1}{2}-\tau}e^{-C\delta d^{2\tau}\log{d}}$, which goes to zero for any $\delta>0$ in the limit of $d\rightarrow\infty$.
Therefore, for a system of sufficiently large dimension $d$, it holds that
\begin{align}
\mathbb{P}_H\left\{\frac{S_{\povm{P}^{(d)}}(U\rho U^\dag)}{\log d}> 1-\delta\right\} \to 1\;,
\end{align}
for any $\delta>0$, in line with what would be expected from a typical, coarse macroscopic observation.

\subsection{Physical random evolutions}
The Haar distribution, while mathematically convenient, is often considered an unphysical model for random evolution, because the amount of randomness required to sample from it, even approximately, grows exponentially with the dimension of the system. To address this problem, following a recent trend in theoretical condensed matter physics -- see, e.g., Refs.~\cite{hayden-preskill-black-mirror,Nahum-PRX-2018} -- we replace the continuous Haar distribution with $\ep$-approximate 2-designs, i.e., finite sets $\mE$ of unitary operators which, if chosen at random, are able to reproduce, up to an error $\ep\ge 0$, many features of the Haar distribution that are of physical interest~\cite{dankert-2009-2-designs,low2009large}. The relevance of approximate 2-designs lies in the fact that it has recently been shown, using rigorous complexity-theoretic arguments~\cite{brandao2021models,harrow2023approximate}, that indeed $\ep$-approximate 2-designs and, more generally, $k$-designs can be efficiently implemented as short random circuits, further justifying them as a physically reasonable model for random evolutions. As shown in Appendix~\ref{appendix:concentration}, we are able to derive a more general law of generic OE increase, valid also for $\ep$-approximate 2-designs.

\begin{theorem}[Approximate 2-design case]\label{th:con1}
    For a unitary operator $U$ sampled at random from an $\ep$-approximate $2$-design $\mE$, regardless of the initial state $\rho$ of the $d$-dimensional system at hand, we have
    \begin{align}
        &\mathbb{P}_{\mE}\left\{S_{\povm{P}}(U\rho U^\dagger)\leq(1-\delta)\log{d}\right\}\nonumber\\
        &\qquad\qquad\leq\frac{1}{\kappa(\povm{P})^{3}d\log{d}}\frac{4(1+\ep)}{\delta}\;,\label{eq:OE-concentration-design}
    \end{align}
    for any value $\delta >0$.
\end{theorem}

The upper bound given in Eq.~\eqref{eq:OE-concentration-design} is weaker with respect to that given in Eq.~\eqref{eq:concentration}, since the negative exponential rate in $d$ that was present in~\eqref{eq:concentration} is now lost, replaced by $(d\log d)^{-1}$. This is due to the fact that the assumptions in Theorem~\ref{th:con1} are weaker, in the sense that an $\ep$-approximate $2$ design is only an approximation of the ideal, but unphysical, Haar distribution. Again, we emphasize that while $\ep$-approximate $2$-designs are physically reasonable because they can be implemented efficiently with little randomness~\cite{brandao2021models,harrow2023approximate}, the Haar distribution is more of a mathematical abstraction. 

Nevertheless, the right-hand side of Eq.~\eqref{eq:OE-concentration-design} is still very small in a large window of parameter values. For example, using the same values of the parameters that we used in the back-of-the-envelope calculation that we did for the Haar-random case, the probability that the OE, after an evolution sampled at random from an $\ep$-approximate 2-design ($\ep<1$), does not exceed 124 bits (when the maximum is 128) is again quite small:
\begin{align*}
\frac{1}{\kappa(\povm{P})^{3}d\log{d}}\frac{4(1+\ep)}{\delta}<\frac{2^{258}}{2^{277}}\times 2^{8}=2^{-11}\;.
\end{align*}
Note that, crucially, $\kappa(\povm{P})$ is again the only parameter that depends on the POVM: everything else is fixed.

More importantly, Eq.~\eqref{eq:OE-concentration-design} still allows us to prove an asymptotic result, although the notion of asymptotic coarseness given in Definition~\ref{def:coarse} must be modified with respect to that obtained before using Eq.~\eqref{eq:concentration}. Now, in the case of $\ep$-approximate 2-designs, we require that $\kappa(d)=\Omega(d^{-\frac{1}{3}+\tau})$ for  some $\tau>0$. Intuitively, this means that von Neumann's condition~\eqref{eq:von-neu-condition} has to be replaced with $d^{2/3}\ll \min_x\tr{P_x}$. In other words, asymptotic coarseness for $\ep$-approximate 2-designs is \textit{coarser} than in Definition~\ref{def:coarse}, which was introduced having in mind the case of Haar-random unitaries. But if the stricter requirement is satisfied, the right-hand side of Eq.~(\ref{eq:OE-concentration-design}) is of order $4(1+\ep)/(\delta d^{3\tau}\log d)$, which goes to zero for any $\ep,\delta>0$ in the limit of $d\rightarrow\infty$. Therefore, for any sufficiently large $d$, even in the case where $U$ is sampled from an $\ep$-approximate $2$-design $\mE$, it still holds that, for any $\delta>0$,
\begin{align}
\mathbb{P}_{\mE}\left\{\frac{S_{\povm{P}^{(d)}}(U\rho U^\dag)}{\log d}> 1-\delta\right\} \to 1\;.
\end{align}

\section{Conclusions}
In this paper we have demonstrated three ways in which observational entropy tends to increase and reach its maximum in isolated systems undergoing a generic unitary evolution.
First, we showed that if the system starts in a non-uniform macroscopic state, only unitary operators belonging to a subvariety of zero volume in the set of all unitaries keep OE invariant, otherwise OE strictly increases. Our proofs here were entirely algebraic, relying on Petz's theory of statistical sufficiency. We then moved on to the problem of showing that the increase of OE, in sufficiently large systems and for sufficiently coarse observations, is a generic phenomenon, independent of the initial state of the system. We considered both Haar-random evolutions, which give better bounds but are not physically reasonable, and $\ep$-approximate 2-designs, which give looser bounds but provide a realistic model of random physical evolutions. In both cases, we found that for sufficiently large systems and coarse observations, the state of the system quickly becomes indistinguishable from the uniform one.

An important point left open for future research is the connection between the random evolution assumption used here and the eigenstate thermalization hypothesis, which has also been shown to lead, under additional physical assumptions, to statements similar to ours~\cite{rigol-srednicki-2012,reimann-2015-PRL-thermalization,strasberg-2023-classicality}. Another open question is to see if it is possible to show OE concentration inequalities for concrete Hamiltonians, such as the free-fermion chain~\cite{tasaki2024macroscopic}, or for random matrix product states and their long time averages~\cite{haferkamp2021emergent}.

\section{Acknowledgments}

We thank Anna Jen\v{c}ov\'a and Philipp Strasberg for insightful discussions. T.~N. acknowledges the ``Nagoya University
Interdisciplinary Frontier Fellowship'' supported by Nagoya University and JST, the establishment of university fellowships towards the creation of science technology innovation, Grant Number JPMJFS2120 and ``THERS Make New Standards Program for the Next Generation Researchers'' supported by JST SPRING, Grant Number JPMJSP2125. K.~K. acknowledges support from JSPS Grant-in-Aid for Early-Career Scientists, No. 22K13972; from MEXT-JSPS Grant-in-Aid for Transformative Research Areas (A) ``Extreme Universe,” No. 22H05254. 
K.~K, E.~W. and F.~B. acknowledge support from MEXT Quantum Leap Flagship Program (MEXT QLEAP) Grant No. JPMXS0120319794.
F.~B. also acknowledges  support from MEXT-JSPS  Grant-in-Aid for Transformative Research Areas  (A) ``Extreme Universe,'' No.~21H05183, and  from JSPS  KAKENHI Grants No.~20K03746 and No.~23K03230.

\bibliography{library}

\begin{thebibliography}{58}%
\makeatletter
\providecommand \@ifxundefined [1]{%
 \@ifx{#1\undefined}
}%
\providecommand \@ifnum [1]{%
 \ifnum #1\expandafter \@firstoftwo
 \else \expandafter \@secondoftwo
 \fi
}%
\providecommand \@ifx [1]{%
 \ifx #1\expandafter \@firstoftwo
 \else \expandafter \@secondoftwo
 \fi
}%
\providecommand \natexlab [1]{#1}%
\providecommand \enquote  [1]{``#1''}%
\providecommand \bibnamefont  [1]{#1}%
\providecommand \bibfnamefont [1]{#1}%
\providecommand \citenamefont [1]{#1}%
\providecommand \href@noop [0]{\@secondoftwo}%
\providecommand \href [0]{\begingroup \@sanitize@url \@href}%
\providecommand \@href[1]{\@@startlink{#1}\@@href}%
\providecommand \@@href[1]{\endgroup#1\@@endlink}%
\providecommand \@sanitize@url [0]{\catcode `\\12\catcode `\$12\catcode `\&12\catcode `\#12\catcode `\^12\catcode `\_12\catcode `\%12\relax}%
\providecommand \@@startlink[1]{}%
\providecommand \@@endlink[0]{}%
\providecommand \url  [0]{\begingroup\@sanitize@url \@url }%
\providecommand \@url [1]{\endgroup\@href {#1}{\urlprefix }}%
\providecommand \urlprefix  [0]{URL }%
\providecommand \Eprint [0]{\href }%
\providecommand \doibase [0]{https://doi.org/}%
\providecommand \selectlanguage [0]{\@gobble}%
\providecommand \bibinfo  [0]{\@secondoftwo}%
\providecommand \bibfield  [0]{\@secondoftwo}%
\providecommand \translation [1]{[#1]}%
\providecommand \BibitemOpen [0]{}%
\providecommand \bibitemStop [0]{}%
\providecommand \bibitemNoStop [0]{.\EOS\space}%
\providecommand \EOS [0]{\spacefactor3000\relax}%
\providecommand \BibitemShut  [1]{\csname bibitem#1\endcsname}%
\let\auto@bib@innerbib\@empty
\bibitem [{\citenamefont {von Neumann}(1955)}]{von1955mathematical}%
  \BibitemOpen
  \bibfield  {author} {\bibinfo {author} {\bibfnamefont {J.}~\bibnamefont {von Neumann}},\ }\href@noop {} {\emph {\bibinfo {title} {Mathematical foundations of quantum mechanics}}}\ (\bibinfo  {publisher} {\href{http://press.princeton.edu/titles/2113.html}{Princeton university press}},\ \bibinfo {year} {1955})\BibitemShut {NoStop}%
\bibitem [{\citenamefont {Wehrl}(1978)}]{wehrl-1978-general-entropy}%
  \BibitemOpen
  \bibfield  {author} {\bibinfo {author} {\bibfnamefont {A.}~\bibnamefont {Wehrl}},\ }\bibfield  {title} {\bibinfo {title} {General properties of entropy},\ }\href {https://doi.org/10.1103/RevModPhys.50.221} {\bibfield  {journal} {\bibinfo  {journal} {Rev. Mod. Phys.}\ }\textbf {\bibinfo {volume} {50}},\ \bibinfo {pages} {221} (\bibinfo {year} {1978})}\BibitemShut {NoStop}%
\bibitem [{\citenamefont {{{\v{S}}afr{\'a}nek}}\ \emph {et~al.}(2019{\natexlab{a}})\citenamefont {{{\v{S}}afr{\'a}nek}}, \citenamefont {{Deutsch}},\ and\ \citenamefont {{Aguirre}}}]{safranek2019a}%
  \BibitemOpen
  \bibfield  {author} {\bibinfo {author} {\bibfnamefont {D.}~\bibnamefont {{{\v{S}}afr{\'a}nek}}}, \bibinfo {author} {\bibfnamefont {J.~M.}\ \bibnamefont {{Deutsch}}},\ and\ \bibinfo {author} {\bibfnamefont {A.}~\bibnamefont {{Aguirre}}},\ }\bibfield  {title} {\bibinfo {title} {{Quantum coarse-grained entropy and thermodynamics}},\ }\href {https://doi.org/10.1103/PhysRevA.99.010101} {\bibfield  {journal} {\bibinfo  {journal} {Phys. Rev. A}\ }\textbf {\bibinfo {volume} {99}},\ \bibinfo {eid} {010101} (\bibinfo {year} {2019}{\natexlab{a}})},\ \Eprint {https://arxiv.org/abs/1707.09722} {arXiv:1707.09722 [quant-ph]} \BibitemShut {NoStop}%
\bibitem [{\citenamefont {{{\v{S}}afr{\'a}nek}}\ \emph {et~al.}(2019{\natexlab{b}})\citenamefont {{{\v{S}}afr{\'a}nek}}, \citenamefont {{Deutsch}},\ and\ \citenamefont {{Aguirre}}}]{safranek2019b}%
  \BibitemOpen
  \bibfield  {author} {\bibinfo {author} {\bibfnamefont {D.}~\bibnamefont {{{\v{S}}afr{\'a}nek}}}, \bibinfo {author} {\bibfnamefont {J.~M.}\ \bibnamefont {{Deutsch}}},\ and\ \bibinfo {author} {\bibfnamefont {A.}~\bibnamefont {{Aguirre}}},\ }\bibfield  {title} {\bibinfo {title} {{Quantum coarse-grained entropy and thermalization in closed systems}},\ }\href {https://doi.org/10.1103/PhysRevA.99.012103} {\bibfield  {journal} {\bibinfo  {journal} {Phys. Rev. A}\ }\textbf {\bibinfo {volume} {99}},\ \bibinfo {eid} {012103} (\bibinfo {year} {2019}{\natexlab{b}})},\ \Eprint {https://arxiv.org/abs/1803.00665} {arXiv:1803.00665 [quant-ph]} \BibitemShut {NoStop}%
\bibitem [{\citenamefont {{{\v{S}}afr{\'a}nek}}\ \emph {et~al.}(2021)\citenamefont {{{\v{S}}afr{\'a}nek}}, \citenamefont {{Aguirre}}, \citenamefont {{Schindler}},\ and\ \citenamefont {{Deutsch}}}]{safranek2021brief}%
  \BibitemOpen
  \bibfield  {author} {\bibinfo {author} {\bibfnamefont {D.}~\bibnamefont {{{\v{S}}afr{\'a}nek}}}, \bibinfo {author} {\bibfnamefont {A.}~\bibnamefont {{Aguirre}}}, \bibinfo {author} {\bibfnamefont {J.}~\bibnamefont {{Schindler}}},\ and\ \bibinfo {author} {\bibfnamefont {J.~M.}\ \bibnamefont {{Deutsch}}},\ }\bibfield  {title} {\bibinfo {title} {{A Brief Introduction to Observational Entropy}},\ }\href {https://doi.org/10.1007/s10701-021-00498-x} {\bibfield  {journal} {\bibinfo  {journal} {Foundations of Physics}\ }\textbf {\bibinfo {volume} {51}},\ \bibinfo {eid} {101} (\bibinfo {year} {2021})},\ \Eprint {https://arxiv.org/abs/2008.04409} {arXiv:2008.04409 [quant-ph]} \BibitemShut {NoStop}%
\bibitem [{\citenamefont {Buscemi}\ \emph {et~al.}(2023)\citenamefont {Buscemi}, \citenamefont {Schindler},\ and\ \citenamefont {{\v{S}}afr{\'{a}}nek}}]{buscemi2022observational}%
  \BibitemOpen
  \bibfield  {author} {\bibinfo {author} {\bibfnamefont {F.}~\bibnamefont {Buscemi}}, \bibinfo {author} {\bibfnamefont {J.}~\bibnamefont {Schindler}},\ and\ \bibinfo {author} {\bibfnamefont {D.}~\bibnamefont {{\v{S}}afr{\'{a}}nek}},\ }\bibfield  {title} {\bibinfo {title} {Observational entropy, coarse-grained states, and the {P}etz recovery map: information-theoretic properties and bounds},\ }\href {https://doi.org/10.1088/1367-2630/accd11} {\bibfield  {journal} {\bibinfo  {journal} {New Journal of Physics}\ }\textbf {\bibinfo {volume} {25}},\ \bibinfo {pages} {053002} (\bibinfo {year} {2023})}\BibitemShut {NoStop}%
\bibitem [{\citenamefont {Bai}\ \emph {et~al.}(2023)\citenamefont {Bai}, \citenamefont {{\v S}afr{\'a}nek}, \citenamefont {Schindler}, \citenamefont {Buscemi},\ and\ \citenamefont {Scarani}}]{bai2023observational}%
  \BibitemOpen
  \bibfield  {author} {\bibinfo {author} {\bibfnamefont {G.}~\bibnamefont {Bai}}, \bibinfo {author} {\bibfnamefont {D.}~\bibnamefont {{\v S}afr{\'a}nek}}, \bibinfo {author} {\bibfnamefont {J.}~\bibnamefont {Schindler}}, \bibinfo {author} {\bibfnamefont {F.}~\bibnamefont {Buscemi}},\ and\ \bibinfo {author} {\bibfnamefont {V.}~\bibnamefont {Scarani}},\ }\href@noop {} {\bibinfo {title} {Observational entropy with general quantum priors}} (\bibinfo {year} {2023}),\ \Eprint {https://arxiv.org/abs/2308.08763} {arXiv:2308.08763 [quant-ph]} \BibitemShut {NoStop}%
\bibitem [{\citenamefont {Strasberg}\ and\ \citenamefont {Winter}(2021)}]{strasberg-winter-2021-PRX-quantum}%
  \BibitemOpen
  \bibfield  {author} {\bibinfo {author} {\bibfnamefont {P.}~\bibnamefont {Strasberg}}\ and\ \bibinfo {author} {\bibfnamefont {A.}~\bibnamefont {Winter}},\ }\bibfield  {title} {\bibinfo {title} {First and second law of quantum thermodynamics: A consistent derivation based on a microscopic definition of entropy},\ }\href {https://doi.org/10.1103/PRXQuantum.2.030202} {\bibfield  {journal} {\bibinfo  {journal} {PRX Quantum}\ }\textbf {\bibinfo {volume} {2}},\ \bibinfo {pages} {030202} (\bibinfo {year} {2021})}\BibitemShut {NoStop}%
\bibitem [{\citenamefont {Riera-Campeny}\ \emph {et~al.}(2021)\citenamefont {Riera-Campeny}, \citenamefont {Sanpera},\ and\ \citenamefont {Strasberg}}]{riera2020finite}%
  \BibitemOpen
  \bibfield  {author} {\bibinfo {author} {\bibfnamefont {A.}~\bibnamefont {Riera-Campeny}}, \bibinfo {author} {\bibfnamefont {A.}~\bibnamefont {Sanpera}},\ and\ \bibinfo {author} {\bibfnamefont {P.}~\bibnamefont {Strasberg}},\ }\bibfield  {title} {\bibinfo {title} {Quantum systems correlated with a finite bath: Nonequilibrium dynamics and thermodynamics},\ }\href {https://doi.org/10.1103/PRXQuantum.2.010340} {\bibfield  {journal} {\bibinfo  {journal} {PRX Quantum}\ }\textbf {\bibinfo {volume} {2}},\ \bibinfo {pages} {010340} (\bibinfo {year} {2021})},\ \Eprint {https://arxiv.org/abs/2008.02184} {arXiv:2008.02184 [quant-ph]} \BibitemShut {NoStop}%
\bibitem [{\citenamefont {{{\v{S}}afr{\'a}nek}}\ \emph {et~al.}(2020)\citenamefont {{{\v{S}}afr{\'a}nek}}, \citenamefont {{Aguirre}},\ and\ \citenamefont {{Deutsch}}}]{safranek2020classical}%
  \BibitemOpen
  \bibfield  {author} {\bibinfo {author} {\bibfnamefont {D.}~\bibnamefont {{{\v{S}}afr{\'a}nek}}}, \bibinfo {author} {\bibfnamefont {A.}~\bibnamefont {{Aguirre}}},\ and\ \bibinfo {author} {\bibfnamefont {J.~M.}\ \bibnamefont {{Deutsch}}},\ }\bibfield  {title} {\bibinfo {title} {{Classical dynamical coarse-grained entropy and comparison with the quantum version}},\ }\href {https://doi.org/10.1103/PhysRevE.102.032106} {\bibfield  {journal} {\bibinfo  {journal} {Phys. Rev. E}\ }\textbf {\bibinfo {volume} {102}},\ \bibinfo {eid} {032106} (\bibinfo {year} {2020})},\ \Eprint {https://arxiv.org/abs/1905.03841} {arXiv:1905.03841 [cond-mat.stat-mech]} \BibitemShut {NoStop}%
\bibitem [{\citenamefont {{Deutsch}}\ \emph {et~al.}(2020)\citenamefont {{Deutsch}}, \citenamefont {{{\v{S}}afr{\'a}nek}},\ and\ \citenamefont {{Aguirre}}}]{deutsch2020probabilistic}%
  \BibitemOpen
  \bibfield  {author} {\bibinfo {author} {\bibfnamefont {J.~M.}\ \bibnamefont {{Deutsch}}}, \bibinfo {author} {\bibfnamefont {D.}~\bibnamefont {{{\v{S}}afr{\'a}nek}}},\ and\ \bibinfo {author} {\bibfnamefont {A.}~\bibnamefont {{Aguirre}}},\ }\bibfield  {title} {\bibinfo {title} {{Probabilistic bound on extreme fluctuations in isolated quantum systems}},\ }\href {https://doi.org/10.1103/PhysRevE.101.032112} {\bibfield  {journal} {\bibinfo  {journal} {\pre}\ }\textbf {\bibinfo {volume} {101}},\ \bibinfo {eid} {032112} (\bibinfo {year} {2020})},\ \Eprint {https://arxiv.org/abs/1806.08897} {arXiv:1806.08897 [gr-qc]} \BibitemShut {NoStop}%
\bibitem [{\citenamefont {{Faiez}}\ \emph {et~al.}(2020)\citenamefont {{Faiez}}, \citenamefont {{{\v{S}}afr{\'a}nek}}, \citenamefont {{Deutsch}},\ and\ \citenamefont {{Aguirre}}}]{faiez2020typical}%
  \BibitemOpen
  \bibfield  {author} {\bibinfo {author} {\bibfnamefont {D.}~\bibnamefont {{Faiez}}}, \bibinfo {author} {\bibfnamefont {D.}~\bibnamefont {{{\v{S}}afr{\'a}nek}}}, \bibinfo {author} {\bibfnamefont {J.~M.}\ \bibnamefont {{Deutsch}}},\ and\ \bibinfo {author} {\bibfnamefont {A.}~\bibnamefont {{Aguirre}}},\ }\bibfield  {title} {\bibinfo {title} {{Typical and extreme entropies of long-lived isolated quantum systems}},\ }\href {https://doi.org/10.1103/PhysRevA.101.052101} {\bibfield  {journal} {\bibinfo  {journal} {\pra}\ }\textbf {\bibinfo {volume} {101}},\ \bibinfo {eid} {052101} (\bibinfo {year} {2020})},\ \Eprint {https://arxiv.org/abs/1908.07083} {arXiv:1908.07083 [quant-ph]} \BibitemShut {NoStop}%
\bibitem [{\citenamefont {{Nation}}\ and\ \citenamefont {{Porras}}(2020)}]{nation2020snapshots}%
  \BibitemOpen
  \bibfield  {author} {\bibinfo {author} {\bibfnamefont {C.}~\bibnamefont {{Nation}}}\ and\ \bibinfo {author} {\bibfnamefont {D.}~\bibnamefont {{Porras}}},\ }\bibfield  {title} {\bibinfo {title} {{Taking snapshots of a quantum thermalization process: Emergent classicality in quantum jump trajectories}},\ }\href {https://doi.org/10.1103/PhysRevE.102.042115} {\bibfield  {journal} {\bibinfo  {journal} {\pre}\ }\textbf {\bibinfo {volume} {102}},\ \bibinfo {eid} {042115} (\bibinfo {year} {2020})},\ \Eprint {https://arxiv.org/abs/2003.08425} {arXiv:2003.08425 [quant-ph]} \BibitemShut {NoStop}%
\bibitem [{\citenamefont {{Strasberg}}\ \emph {et~al.}(2021)\citenamefont {{Strasberg}}, \citenamefont {{D{\'\i}az}},\ and\ \citenamefont {{Riera-Campeny}}}]{strasberg2021clausius}%
  \BibitemOpen
  \bibfield  {author} {\bibinfo {author} {\bibfnamefont {P.}~\bibnamefont {{Strasberg}}}, \bibinfo {author} {\bibfnamefont {M.~G.}\ \bibnamefont {{D{\'\i}az}}},\ and\ \bibinfo {author} {\bibfnamefont {A.}~\bibnamefont {{Riera-Campeny}}},\ }\bibfield  {title} {\bibinfo {title} {{Clausius inequality for finite baths reveals universal efficiency improvements}},\ }\href {https://doi.org/10.1103/PhysRevE.104.L022103} {\bibfield  {journal} {\bibinfo  {journal} {Phys. Rev. E}\ }\textbf {\bibinfo {volume} {104}},\ \bibinfo {eid} {L022103} (\bibinfo {year} {2021})},\ \Eprint {https://arxiv.org/abs/2012.03262} {arXiv:2012.03262 [quant-ph]} \BibitemShut {NoStop}%
\bibitem [{\citenamefont {{Hamazaki}}(2022)}]{hamazaki2022speed}%
  \BibitemOpen
  \bibfield  {author} {\bibinfo {author} {\bibfnamefont {R.}~\bibnamefont {{Hamazaki}}},\ }\bibfield  {title} {\bibinfo {title} {{Speed Limits for Macroscopic Transitions}},\ }\href {https://doi.org/10.1103/PRXQuantum.3.020319} {\bibfield  {journal} {\bibinfo  {journal} {PRX Quantum}\ }\textbf {\bibinfo {volume} {3}},\ \bibinfo {eid} {020319} (\bibinfo {year} {2022})},\ \Eprint {https://arxiv.org/abs/2110.09716} {arXiv:2110.09716 [cond-mat.stat-mech]} \BibitemShut {NoStop}%
\bibitem [{\citenamefont {Modak}\ and\ \citenamefont {Aravinda}(2022)}]{modak2022observational}%
  \BibitemOpen
  \bibfield  {author} {\bibinfo {author} {\bibfnamefont {R.}~\bibnamefont {Modak}}\ and\ \bibinfo {author} {\bibfnamefont {S.}~\bibnamefont {Aravinda}},\ }\bibfield  {title} {\bibinfo {title} {Observational-entropic study of anderson localization},\ }\href {https://doi.org/10.1103/PhysRevA.106.062217} {\bibfield  {journal} {\bibinfo  {journal} {Phys. Rev. A}\ }\textbf {\bibinfo {volume} {106}},\ \bibinfo {pages} {062217} (\bibinfo {year} {2022})}\BibitemShut {NoStop}%
\bibitem [{\citenamefont {PG}\ \emph {et~al.}(2023)\citenamefont {PG}, \citenamefont {Modak},\ and\ \citenamefont {Aravinda}}]{sreeram2023witnessing}%
  \BibitemOpen
  \bibfield  {author} {\bibinfo {author} {\bibfnamefont {S.}~\bibnamefont {PG}}, \bibinfo {author} {\bibfnamefont {R.}~\bibnamefont {Modak}},\ and\ \bibinfo {author} {\bibfnamefont {S.}~\bibnamefont {Aravinda}},\ }\bibfield  {title} {\bibinfo {title} {Witnessing quantum chaos using observational entropy},\ }\href {https://doi.org/10.1103/PhysRevE.107.064204} {\bibfield  {journal} {\bibinfo  {journal} {Phys. Rev. E}\ }\textbf {\bibinfo {volume} {107}},\ \bibinfo {pages} {064204} (\bibinfo {year} {2023})}\BibitemShut {NoStop}%
\bibitem [{\citenamefont {{Schindler}}\ and\ \citenamefont {{Winter}}(2023)}]{schindler2023continuity}%
  \BibitemOpen
  \bibfield  {author} {\bibinfo {author} {\bibfnamefont {J.}~\bibnamefont {{Schindler}}}\ and\ \bibinfo {author} {\bibfnamefont {A.}~\bibnamefont {{Winter}}},\ }\bibfield  {title} {\bibinfo {title} {{Continuity bounds on observational entropy and measured relative entropies}},\ }\href {https://doi.org/10.48550/arXiv.2302.00400} {\bibfield  {journal} {\bibinfo  {journal} {arXiv e-prints}\ ,\ \bibinfo {eid} {arXiv:2302.00400}} (\bibinfo {year} {2023})},\ \Eprint {https://arxiv.org/abs/2302.00400} {arXiv:2302.00400 [quant-ph]} \BibitemShut {NoStop}%
\bibitem [{\citenamefont {{{\v{S}}afr{\'{a}}nek}}\ \emph {et~al.}(2023)\citenamefont {{{\v{S}}afr{\'{a}}nek}}, \citenamefont {Rosa},\ and\ \citenamefont {Binder}}]{safranek2023work}%
  \BibitemOpen
  \bibfield  {author} {\bibinfo {author} {\bibfnamefont {D.}~\bibnamefont {{{\v{S}}afr{\'{a}}nek}}}, \bibinfo {author} {\bibfnamefont {D.}~\bibnamefont {Rosa}},\ and\ \bibinfo {author} {\bibfnamefont {F.~C.}\ \bibnamefont {Binder}},\ }\bibfield  {title} {\bibinfo {title} {Work extraction from unknown quantum sources},\ }\href {https://doi.org/10.1103/PhysRevLett.130.210401} {\bibfield  {journal} {\bibinfo  {journal} {Phys. Rev. Lett.}\ }\textbf {\bibinfo {volume} {130}},\ \bibinfo {pages} {210401} (\bibinfo {year} {2023})}\BibitemShut {NoStop}%
\bibitem [{\citenamefont {{{\v{S}}afr{\'{a}}nek}}\ and\ \citenamefont {{Rosa}}(2023)}]{safranek2023measuring}%
  \BibitemOpen
  \bibfield  {author} {\bibinfo {author} {\bibfnamefont {D.}~\bibnamefont {{{\v{S}}afr{\'{a}}nek}}}\ and\ \bibinfo {author} {\bibfnamefont {D.}~\bibnamefont {{Rosa}}},\ }\bibfield  {title} {\bibinfo {title} {Measuring energy by measuring any other observable},\ }\href {https://doi.org/10.1103/PhysRevA.108.022208} {\bibfield  {journal} {\bibinfo  {journal} {Phys. Rev. A}\ }\textbf {\bibinfo {volume} {108}},\ \bibinfo {pages} {022208} (\bibinfo {year} {2023})}\BibitemShut {NoStop}%
\bibitem [{\citenamefont {{{\v{S}}afr{\'a}nek}}(2023)}]{safranek2023ergotropic}%
  \BibitemOpen
  \bibfield  {author} {\bibinfo {author} {\bibfnamefont {D.}~\bibnamefont {{{\v{S}}afr{\'a}nek}}},\ }\bibfield  {title} {\bibinfo {title} {{Ergotropic interpretation of entanglement entropy}},\ }\href {https://doi.org/10.48550/arXiv.2306.08987} {\bibfield  {journal} {\bibinfo  {journal} {arXiv e-prints}\ ,\ \bibinfo {eid} {arXiv:2306.08987}} (\bibinfo {year} {2023})},\ \Eprint {https://arxiv.org/abs/2306.08987} {arXiv:2306.08987 [quant-ph]} \BibitemShut {NoStop}%
\bibitem [{\citenamefont {von Neumann}(2010)}]{vonNeumann1929translation}%
  \BibitemOpen
  \bibfield  {author} {\bibinfo {author} {\bibfnamefont {J.}~\bibnamefont {von Neumann}},\ }\bibfield  {title} {\bibinfo {title} {{Proof of the ergodic theorem and the H-theorem in quantum mechanics. Translation of: Beweis des Ergodensatzes und des H-Theorems in der neuen Mechanik}},\ }\href {https://doi.org/10.1140/epjh/e2010-00008-5} {\bibfield  {journal} {\bibinfo  {journal} {European Physical Journal H}\ }\textbf {\bibinfo {volume} {35}},\ \bibinfo {pages} {201} (\bibinfo {year} {2010})}\BibitemShut {NoStop}%
\bibitem [{\citenamefont {Tolman}(2010)}]{Tolman-2010-dover}%
  \BibitemOpen
  \bibfield  {author} {\bibinfo {author} {\bibfnamefont {R.~C.}\ \bibnamefont {Tolman}},\ }\href@noop {} {\emph {\bibinfo {title} {The Principles of Statistical Mechanics}}}\ (\bibinfo  {publisher} {Dover Publications},\ \bibinfo {year} {2010})\BibitemShut {NoStop}%
\bibitem [{\citenamefont {Goldstein}\ \emph {et~al.}(2010{\natexlab{a}})\citenamefont {Goldstein}, \citenamefont {Lebowitz}, \citenamefont {Tumulka},\ and\ \citenamefont {Zangh{\`\i}}}]{Goldstein2010Commentary}%
  \BibitemOpen
  \bibfield  {author} {\bibinfo {author} {\bibfnamefont {S.}~\bibnamefont {Goldstein}}, \bibinfo {author} {\bibfnamefont {J.~L.}\ \bibnamefont {Lebowitz}}, \bibinfo {author} {\bibfnamefont {R.}~\bibnamefont {Tumulka}},\ and\ \bibinfo {author} {\bibfnamefont {N.}~\bibnamefont {Zangh{\`\i}}},\ }\bibfield  {title} {\bibinfo {title} {Long-time behavior of macroscopic quantum systems: Commentary accompanying the english translation of john von neumann's 1929 article on the quantum ergodic theorem},\ }\href {https://doi.org/10.1140/epjh/e2010-00007-7} {\bibfield  {journal} {\bibinfo  {journal} {The European Physical Journal H}\ }\textbf {\bibinfo {volume} {35}},\ \bibinfo {pages} {173} (\bibinfo {year} {2010}{\natexlab{a}})}\BibitemShut {NoStop}%
\bibitem [{\citenamefont {Goldstein}\ \emph {et~al.}(2010{\natexlab{b}})\citenamefont {Goldstein}, \citenamefont {Lebowitz}, \citenamefont {Mastrodonato}, \citenamefont {Tumulka},\ and\ \citenamefont {Zangh{\`\i}}}]{Goldstein2010NormalTypicality}%
  \BibitemOpen
  \bibfield  {author} {\bibinfo {author} {\bibfnamefont {S.}~\bibnamefont {Goldstein}}, \bibinfo {author} {\bibfnamefont {J.~L.}\ \bibnamefont {Lebowitz}}, \bibinfo {author} {\bibfnamefont {C.}~\bibnamefont {Mastrodonato}}, \bibinfo {author} {\bibfnamefont {R.}~\bibnamefont {Tumulka}},\ and\ \bibinfo {author} {\bibfnamefont {N.}~\bibnamefont {Zangh{\`\i}}},\ }\bibfield  {title} {\bibinfo {title} {Normal typicality and von neumann's quantum ergodic theorem},\ }\href {https://doi.org/10.1098/rspa.2009.0635} {\bibfield  {journal} {\bibinfo  {journal} {Proceedings of the Royal Society A: Mathematical, Physical and Engineering Sciences}\ }\textbf {\bibinfo {volume} {466}},\ \bibinfo {pages} {3203} (\bibinfo {year} {2010}{\natexlab{b}})}\BibitemShut {NoStop}%
\bibitem [{\citenamefont {Petz}(1986)}]{petz1986sufficient}%
  \BibitemOpen
  \bibfield  {author} {\bibinfo {author} {\bibfnamefont {D.}~\bibnamefont {Petz}},\ }\bibfield  {title} {\bibinfo {title} {Sufficient subalgebras and the relative entropy of states of a von neumann algebra},\ }\href {https://doi.org/https://doi.org/10.1007/BF01212345} {\bibfield  {journal} {\bibinfo  {journal} {Communications in mathematical physics}\ }\textbf {\bibinfo {volume} {105}},\ \bibinfo {pages} {123} (\bibinfo {year} {1986})}\BibitemShut {NoStop}%
\bibitem [{\citenamefont {Petz}(1988)}]{petz1988sufficiency}%
  \BibitemOpen
  \bibfield  {author} {\bibinfo {author} {\bibfnamefont {D.}~\bibnamefont {Petz}},\ }\bibfield  {title} {\bibinfo {title} {Sufficiency of channels over von neumann algebras},\ }\href {https://doi.org/https://doi.org/10.1093/qmath/39.1.97} {\bibfield  {journal} {\bibinfo  {journal} {The Quarterly Journal of Mathematics}\ }\textbf {\bibinfo {volume} {39}},\ \bibinfo {pages} {97} (\bibinfo {year} {1988})}\BibitemShut {NoStop}%
\bibitem [{\citenamefont {Petz}(2003)}]{petz2003monotonicity}%
  \BibitemOpen
  \bibfield  {author} {\bibinfo {author} {\bibfnamefont {D.}~\bibnamefont {Petz}},\ }\bibfield  {title} {\bibinfo {title} {Monotonicity of quantum relative entropy revisited},\ }\href {https://doi.org/https://doi.org/10.1142/S0129055X03001576} {\bibfield  {journal} {\bibinfo  {journal} {Reviews in Mathematical Physics}\ }\textbf {\bibinfo {volume} {15}},\ \bibinfo {pages} {79} (\bibinfo {year} {2003})}\BibitemShut {NoStop}%
\bibitem [{\citenamefont {Jen{\v c}ov{\'a}}\ and\ \citenamefont {Petz}(2006)}]{jencova-petz-2006-sufficiency-survey}%
  \BibitemOpen
  \bibfield  {author} {\bibinfo {author} {\bibfnamefont {A.}~\bibnamefont {Jen{\v c}ov{\'a}}}\ and\ \bibinfo {author} {\bibfnamefont {D.}~\bibnamefont {Petz}},\ }\bibfield  {title} {\bibinfo {title} {Sufficiency in quantum statistical inference: a survey with examples},\ }\href {https://doi.org/10.1142/s0219025706002408} {\bibfield  {journal} {\bibinfo  {journal} {Infinite Dimensional Analysis, Quantum Probability and Related Topics}\ }\textbf {\bibinfo {volume} {09}},\ \bibinfo {pages} {331} (\bibinfo {year} {2006})}\BibitemShut {NoStop}%
\bibitem [{\citenamefont {Ledoux}(2001)}]{ledoux-2001}%
  \BibitemOpen
  \bibfield  {author} {\bibinfo {author} {\bibfnamefont {M.}~\bibnamefont {Ledoux}},\ }\href@noop {} {\emph {\bibinfo {title} {The concentration of measure phenomenon}}}\ (\bibinfo  {publisher} {American Mathematical Society},\ \bibinfo {year} {2001})\BibitemShut {NoStop}%
\bibitem [{\citenamefont {Hayden}\ \emph {et~al.}(2006)\citenamefont {Hayden}, \citenamefont {Leung},\ and\ \citenamefont {Winter}}]{hayden-2006-aspects-generic-ent}%
  \BibitemOpen
  \bibfield  {author} {\bibinfo {author} {\bibfnamefont {P.}~\bibnamefont {Hayden}}, \bibinfo {author} {\bibfnamefont {D.~W.}\ \bibnamefont {Leung}},\ and\ \bibinfo {author} {\bibfnamefont {A.}~\bibnamefont {Winter}},\ }\bibfield  {title} {\bibinfo {title} {Aspects of generic entanglement},\ }\href {https://doi.org/10.1007/s00220-006-1535-6} {\bibfield  {journal} {\bibinfo  {journal} {Communications in Mathematical Physics}\ }\textbf {\bibinfo {volume} {265}},\ \bibinfo {pages} {95} (\bibinfo {year} {2006})}\BibitemShut {NoStop}%
\bibitem [{\citenamefont {Kaneko}\ \emph {et~al.}(2020)\citenamefont {Kaneko}, \citenamefont {Iyoda},\ and\ \citenamefont {Sagawa}}]{kaneko2020characterizing}%
  \BibitemOpen
  \bibfield  {author} {\bibinfo {author} {\bibfnamefont {K.}~\bibnamefont {Kaneko}}, \bibinfo {author} {\bibfnamefont {E.}~\bibnamefont {Iyoda}},\ and\ \bibinfo {author} {\bibfnamefont {T.}~\bibnamefont {Sagawa}},\ }\bibfield  {title} {\bibinfo {title} {Characterizing complexity of many-body quantum dynamics by higher-order eigenstate thermalization},\ }\bibfield  {journal} {\bibinfo  {journal} {Physical Review A}\ }\textbf {\bibinfo {volume} {101}},\ \href {https://doi.org/10.1103/physreva.101.042126} {10.1103/physreva.101.042126} (\bibinfo {year} {2020})\BibitemShut {NoStop}%
\bibitem [{\citenamefont {Schuster}\ \emph {et~al.}(2024)\citenamefont {Schuster}, \citenamefont {Haferkamp},\ and\ \citenamefont {Huang}}]{schuster2024random}%
  \BibitemOpen
  \bibfield  {author} {\bibinfo {author} {\bibfnamefont {T.}~\bibnamefont {Schuster}}, \bibinfo {author} {\bibfnamefont {J.}~\bibnamefont {Haferkamp}},\ and\ \bibinfo {author} {\bibfnamefont {H.-Y.}\ \bibnamefont {Huang}},\ }\bibfield  {title} {\bibinfo {title} {Random unitaries in extremely low depth},\ }\href@noop {} {\bibfield  {journal} {\bibinfo  {journal} {arXiv preprint arXiv:2407.07754}\ } (\bibinfo {year} {2024})}\BibitemShut {NoStop}%
\bibitem [{\citenamefont {LaRacuente}\ and\ \citenamefont {Leditzky}(2024)}]{laracuente2024approximate}%
  \BibitemOpen
  \bibfield  {author} {\bibinfo {author} {\bibfnamefont {N.}~\bibnamefont {LaRacuente}}\ and\ \bibinfo {author} {\bibfnamefont {F.}~\bibnamefont {Leditzky}},\ }\bibfield  {title} {\bibinfo {title} {Approximate unitary $ k $-designs from shallow, low-communication circuits},\ }\href@noop {} {\bibfield  {journal} {\bibinfo  {journal} {arXiv preprint arXiv:2407.07876}\ } (\bibinfo {year} {2024})}\BibitemShut {NoStop}%
\bibitem [{\citenamefont {Dankert}\ \emph {et~al.}(2009)\citenamefont {Dankert}, \citenamefont {Cleve}, \citenamefont {Emerson},\ and\ \citenamefont {Livine}}]{dankert-2009-2-designs}%
  \BibitemOpen
  \bibfield  {author} {\bibinfo {author} {\bibfnamefont {C.}~\bibnamefont {Dankert}}, \bibinfo {author} {\bibfnamefont {R.}~\bibnamefont {Cleve}}, \bibinfo {author} {\bibfnamefont {J.}~\bibnamefont {Emerson}},\ and\ \bibinfo {author} {\bibfnamefont {E.}~\bibnamefont {Livine}},\ }\bibfield  {title} {\bibinfo {title} {Exact and approximate unitary 2-designs and their application to fidelity estimation},\ }\href {https://doi.org/10.1103/PhysRevA.80.012304} {\bibfield  {journal} {\bibinfo  {journal} {Phys. Rev. A}\ }\textbf {\bibinfo {volume} {80}},\ \bibinfo {pages} {012304} (\bibinfo {year} {2009})}\BibitemShut {NoStop}%
\bibitem [{\citenamefont {Low}(2009)}]{low2009large}%
  \BibitemOpen
  \bibfield  {author} {\bibinfo {author} {\bibfnamefont {R.~A.}\ \bibnamefont {Low}},\ }\bibfield  {title} {\bibinfo {title} {Large deviation bounds for k-designs},\ }\href {https://doi.org/10.1098/rspa.2009.0232} {\bibfield  {journal} {\bibinfo  {journal} {Proceedings of the Royal Society A: Mathematical, Physical and Engineering Sciences}\ }\textbf {\bibinfo {volume} {465}},\ \bibinfo {pages} {3289} (\bibinfo {year} {2009})}\BibitemShut {NoStop}%
\bibitem [{\citenamefont {Brandão}\ \emph {et~al.}(2021)\citenamefont {Brandão}, \citenamefont {Chemissany}, \citenamefont {Hunter-Jones}, \citenamefont {Kueng},\ and\ \citenamefont {Preskill}}]{brandao2021models}%
  \BibitemOpen
  \bibfield  {author} {\bibinfo {author} {\bibfnamefont {F.~G.}\ \bibnamefont {Brandão}}, \bibinfo {author} {\bibfnamefont {W.}~\bibnamefont {Chemissany}}, \bibinfo {author} {\bibfnamefont {N.}~\bibnamefont {Hunter-Jones}}, \bibinfo {author} {\bibfnamefont {R.}~\bibnamefont {Kueng}},\ and\ \bibinfo {author} {\bibfnamefont {J.}~\bibnamefont {Preskill}},\ }\bibfield  {title} {\bibinfo {title} {Models of quantum complexity growth},\ }\bibfield  {journal} {\bibinfo  {journal} {PRX Quantum}\ }\textbf {\bibinfo {volume} {2}},\ \href {https://doi.org/10.1103/prxquantum.2.030316} {10.1103/prxquantum.2.030316} (\bibinfo {year} {2021})\BibitemShut {NoStop}%
\bibitem [{\citenamefont {Harrow}\ and\ \citenamefont {Mehraban}(2023)}]{harrow2023approximate}%
  \BibitemOpen
  \bibfield  {author} {\bibinfo {author} {\bibfnamefont {A.~W.}\ \bibnamefont {Harrow}}\ and\ \bibinfo {author} {\bibfnamefont {S.}~\bibnamefont {Mehraban}},\ }\bibfield  {title} {\bibinfo {title} {Approximate unitary t-designs by short random quantum circuits using nearest-neighbor and long-range gates},\ }\href {https://doi.org/10.1007/s00220-023-04675-z} {\bibfield  {journal} {\bibinfo  {journal} {Communications in Mathematical Physics}\ }\textbf {\bibinfo {volume} {401}},\ \bibinfo {pages} {1531–1626} (\bibinfo {year} {2023})}\BibitemShut {NoStop}%
\bibitem [{\citenamefont {Nielsen}\ and\ \citenamefont {Chuang}(2010)}]{nielsen_chuang_2010}%
  \BibitemOpen
  \bibfield  {author} {\bibinfo {author} {\bibfnamefont {M.~A.}\ \bibnamefont {Nielsen}}\ and\ \bibinfo {author} {\bibfnamefont {I.~L.}\ \bibnamefont {Chuang}},\ }\href {https://doi.org/https://doi.org/10.1017/CBO9780511976667} {\emph {\bibinfo {title} {Quantum Computation and Quantum Information: 10th Anniversary Edition}}}\ (\bibinfo  {publisher} {Cambridge University Press},\ \bibinfo {year} {2010})\BibitemShut {NoStop}%
\bibitem [{\citenamefont {Wilde}(2013)}]{wilde_2013}%
  \BibitemOpen
  \bibfield  {author} {\bibinfo {author} {\bibfnamefont {M.~M.}\ \bibnamefont {Wilde}},\ }\href {https://doi.org/10.1017/CBO9781139525343} {\emph {\bibinfo {title} {Quantum Information Theory}}}\ (\bibinfo  {publisher} {Cambridge University Press},\ \bibinfo {year} {2013})\BibitemShut {NoStop}%
\bibitem [{\citenamefont {Umegaki}(1961)}]{umegaki-q-rel-ent-1961}%
  \BibitemOpen
  \bibfield  {author} {\bibinfo {author} {\bibfnamefont {H.}~\bibnamefont {Umegaki}},\ }\bibfield  {title} {\bibinfo {title} {On information in operator algebras},\ }\href {https://doi.org/10.3792/pja/1195523632} {\bibfield  {journal} {\bibinfo  {journal} {Proc. Japan Acad.}\ }\textbf {\bibinfo {volume} {37}},\ \bibinfo {pages} {459} (\bibinfo {year} {1961})}\BibitemShut {NoStop}%
\bibitem [{\citenamefont {Umegaki}(1962)}]{umegaki1962conditional}%
  \BibitemOpen
  \bibfield  {author} {\bibinfo {author} {\bibfnamefont {H.}~\bibnamefont {Umegaki}},\ }\bibfield  {title} {\bibinfo {title} {Conditional expectation in an operator algebra, iv (entropy and information)},\ }\href@noop {} {\bibfield  {journal} {\bibinfo  {journal} {Kodai Mathematical Seminar Reports}\ }\textbf {\bibinfo {volume} {14}},\ \bibinfo {pages} {59} (\bibinfo {year} {1962})}\BibitemShut {NoStop}%
\bibitem [{\citenamefont {Kullback}\ and\ \citenamefont {Leibler}(1951)}]{kullback1951information}%
  \BibitemOpen
  \bibfield  {author} {\bibinfo {author} {\bibfnamefont {S.}~\bibnamefont {Kullback}}\ and\ \bibinfo {author} {\bibfnamefont {R.~A.}\ \bibnamefont {Leibler}},\ }\bibfield  {title} {\bibinfo {title} {{On information and sufficiency}},\ }\href {https://projecteuclid.org/euclid.aoms/1177729694} {\bibfield  {journal} {\bibinfo  {journal} {{Ann. Math. Stat.}}\ }\textbf {\bibinfo {volume} {22}},\ \bibinfo {pages} {79} (\bibinfo {year} {1951})}\BibitemShut {NoStop}%
\bibitem [{\citenamefont {{Martens}}\ and\ \citenamefont {{de Muynck}}(1990)}]{martens1990nonideal}%
  \BibitemOpen
  \bibfield  {author} {\bibinfo {author} {\bibfnamefont {H.}~\bibnamefont {{Martens}}}\ and\ \bibinfo {author} {\bibfnamefont {W.~M.}\ \bibnamefont {{de Muynck}}},\ }\bibfield  {title} {\bibinfo {title} {{Nonideal quantum measurements}},\ }\href {https://doi.org/10.1007/BF00731693} {\bibfield  {journal} {\bibinfo  {journal} {Foundations of Physics}\ }\textbf {\bibinfo {volume} {20}},\ \bibinfo {pages} {255} (\bibinfo {year} {1990})}\BibitemShut {NoStop}%
\bibitem [{\citenamefont {Buscemi}\ \emph {et~al.}(2005)\citenamefont {Buscemi}, \citenamefont {Keyl}, \citenamefont {D'Ariano}, \citenamefont {Perinotti},\ and\ \citenamefont {Werner}}]{buscemi-2005-clean-POVMs}%
  \BibitemOpen
  \bibfield  {author} {\bibinfo {author} {\bibfnamefont {F.}~\bibnamefont {Buscemi}}, \bibinfo {author} {\bibfnamefont {M.}~\bibnamefont {Keyl}}, \bibinfo {author} {\bibfnamefont {G.~M.}\ \bibnamefont {D'Ariano}}, \bibinfo {author} {\bibfnamefont {P.}~\bibnamefont {Perinotti}},\ and\ \bibinfo {author} {\bibfnamefont {R.~F.}\ \bibnamefont {Werner}},\ }\bibfield  {title} {\bibinfo {title} {Clean positive operator valued measures},\ }\href {https://doi.org/10.1063/1.2008996} {\bibfield  {journal} {\bibinfo  {journal} {Journal of Mathematical Physics}\ }\textbf {\bibinfo {volume} {46}},\ \bibinfo {pages} {082109} (\bibinfo {year} {2005})},\ \Eprint {https://arxiv.org/abs/https://doi.org/10.1063/1.2008996} {https://doi.org/10.1063/1.2008996} \BibitemShut {NoStop}%
\bibitem [{\citenamefont {Polkovnikov}(2011)}]{polkovnikov2011microscopic}%
  \BibitemOpen
  \bibfield  {author} {\bibinfo {author} {\bibfnamefont {A.}~\bibnamefont {Polkovnikov}},\ }\bibfield  {title} {\bibinfo {title} {Microscopic diagonal entropy and its connection to basic thermodynamic relations},\ }\href {https://doi.org/10.1016/j.aop.2010.08.004} {\bibfield  {journal} {\bibinfo  {journal} {Annals of Physics}\ }\textbf {\bibinfo {volume} {326}},\ \bibinfo {pages} {486–499} (\bibinfo {year} {2011})}\BibitemShut {NoStop}%
\bibitem [{\citenamefont {Buscemi}\ \emph {et~al.}(2020)\citenamefont {Buscemi}, \citenamefont {Fujiwara}, \citenamefont {Mitsui},\ and\ \citenamefont {Rotondo}}]{buscemi2020thermodynamic}%
  \BibitemOpen
  \bibfield  {author} {\bibinfo {author} {\bibfnamefont {F.}~\bibnamefont {Buscemi}}, \bibinfo {author} {\bibfnamefont {D.}~\bibnamefont {Fujiwara}}, \bibinfo {author} {\bibfnamefont {N.}~\bibnamefont {Mitsui}},\ and\ \bibinfo {author} {\bibfnamefont {M.}~\bibnamefont {Rotondo}},\ }\bibfield  {title} {\bibinfo {title} {Thermodynamic reverse bounds for general open quantum processes},\ }\href@noop {} {\bibfield  {journal} {\bibinfo  {journal} {Physical Review A}\ }\textbf {\bibinfo {volume} {102}},\ \bibinfo {pages} {032210} (\bibinfo {year} {2020})}\BibitemShut {NoStop}%
\bibitem [{\citenamefont {Buscemi}\ and\ \citenamefont {Scarani}(2021)}]{buscemi-scarani-2021fluctuation}%
  \BibitemOpen
  \bibfield  {author} {\bibinfo {author} {\bibfnamefont {F.}~\bibnamefont {Buscemi}}\ and\ \bibinfo {author} {\bibfnamefont {V.}~\bibnamefont {Scarani}},\ }\bibfield  {title} {\bibinfo {title} {Fluctuation theorems from bayesian retrodiction},\ }\href {https://doi.org/https://doi.org/10.1103/PhysRevE.103.052111} {\bibfield  {journal} {\bibinfo  {journal} {Phys. Rev. E}\ }\textbf {\bibinfo {volume} {103}},\ \bibinfo {pages} {052111} (\bibinfo {year} {2021})}\BibitemShut {NoStop}%
\bibitem [{\citenamefont {Aw}\ \emph {et~al.}(2021)\citenamefont {Aw}, \citenamefont {Buscemi},\ and\ \citenamefont {Scarani}}]{aw-buscemi-scarani}%
  \BibitemOpen
  \bibfield  {author} {\bibinfo {author} {\bibfnamefont {C.~C.}\ \bibnamefont {Aw}}, \bibinfo {author} {\bibfnamefont {F.}~\bibnamefont {Buscemi}},\ and\ \bibinfo {author} {\bibfnamefont {V.}~\bibnamefont {Scarani}},\ }\bibfield  {title} {\bibinfo {title} {Fluctuation theorems with retrodiction rather than reverse processes},\ }\href {https://doi.org/https://doi.org/10.1116/5.0060893} {\bibfield  {journal} {\bibinfo  {journal} {AVS Quantum Science}\ }\textbf {\bibinfo {volume} {3}},\ \bibinfo {pages} {045601} (\bibinfo {year} {2021})}\BibitemShut {NoStop}%
\bibitem [{\citenamefont {Hayden}\ and\ \citenamefont {Preskill}(2007)}]{hayden-preskill-black-mirror}%
  \BibitemOpen
  \bibfield  {author} {\bibinfo {author} {\bibfnamefont {P.}~\bibnamefont {Hayden}}\ and\ \bibinfo {author} {\bibfnamefont {J.}~\bibnamefont {Preskill}},\ }\bibfield  {title} {\bibinfo {title} {Black holes as mirrors: quantum information in random subsystems},\ }\href {https://doi.org/10.1088/1126-6708/2007/09/120} {\bibfield  {journal} {\bibinfo  {journal} {Journal of High Energy Physics}\ }\textbf {\bibinfo {volume} {2007}},\ \bibinfo {pages} {120} (\bibinfo {year} {2007})}\BibitemShut {NoStop}%
\bibitem [{\citenamefont {Nahum}\ \emph {et~al.}(2018)\citenamefont {Nahum}, \citenamefont {Vijay},\ and\ \citenamefont {Haah}}]{Nahum-PRX-2018}%
  \BibitemOpen
  \bibfield  {author} {\bibinfo {author} {\bibfnamefont {A.}~\bibnamefont {Nahum}}, \bibinfo {author} {\bibfnamefont {S.}~\bibnamefont {Vijay}},\ and\ \bibinfo {author} {\bibfnamefont {J.}~\bibnamefont {Haah}},\ }\bibfield  {title} {\bibinfo {title} {Operator spreading in random unitary circuits},\ }\href {https://doi.org/10.1103/PhysRevX.8.021014} {\bibfield  {journal} {\bibinfo  {journal} {Phys. Rev. X}\ }\textbf {\bibinfo {volume} {8}},\ \bibinfo {pages} {021014} (\bibinfo {year} {2018})}\BibitemShut {NoStop}%
\bibitem [{\citenamefont {Rigol}\ and\ \citenamefont {Srednicki}(2012)}]{rigol-srednicki-2012}%
  \BibitemOpen
  \bibfield  {author} {\bibinfo {author} {\bibfnamefont {M.}~\bibnamefont {Rigol}}\ and\ \bibinfo {author} {\bibfnamefont {M.}~\bibnamefont {Srednicki}},\ }\bibfield  {title} {\bibinfo {title} {Alternatives to eigenstate thermalization},\ }\href {https://doi.org/10.1103/PhysRevLett.108.110601} {\bibfield  {journal} {\bibinfo  {journal} {Phys. Rev. Lett.}\ }\textbf {\bibinfo {volume} {108}},\ \bibinfo {pages} {110601} (\bibinfo {year} {2012})}\BibitemShut {NoStop}%
\bibitem [{\citenamefont {Reimann}(2015)}]{reimann-2015-PRL-thermalization}%
  \BibitemOpen
  \bibfield  {author} {\bibinfo {author} {\bibfnamefont {P.}~\bibnamefont {Reimann}},\ }\bibfield  {title} {\bibinfo {title} {Generalization of von neumann's approach to thermalization},\ }\href {https://doi.org/10.1103/PhysRevLett.115.010403} {\bibfield  {journal} {\bibinfo  {journal} {Phys. Rev. Lett.}\ }\textbf {\bibinfo {volume} {115}},\ \bibinfo {pages} {010403} (\bibinfo {year} {2015})}\BibitemShut {NoStop}%
\bibitem [{\citenamefont {Strasberg}\ \emph {et~al.}(2023)\citenamefont {Strasberg}, \citenamefont {Winter}, \citenamefont {Gemmer},\ and\ \citenamefont {Wang}}]{strasberg-2023-classicality}%
  \BibitemOpen
  \bibfield  {author} {\bibinfo {author} {\bibfnamefont {P.}~\bibnamefont {Strasberg}}, \bibinfo {author} {\bibfnamefont {A.}~\bibnamefont {Winter}}, \bibinfo {author} {\bibfnamefont {J.}~\bibnamefont {Gemmer}},\ and\ \bibinfo {author} {\bibfnamefont {J.}~\bibnamefont {Wang}},\ }\bibfield  {title} {\bibinfo {title} {Classicality, markovianity, and local detailed balance from pure-state dynamics},\ }\href {https://doi.org/10.1103/PhysRevA.108.012225} {\bibfield  {journal} {\bibinfo  {journal} {Phys. Rev. A}\ }\textbf {\bibinfo {volume} {108}},\ \bibinfo {pages} {012225} (\bibinfo {year} {2023})}\BibitemShut {NoStop}%
\bibitem [{\citenamefont {Tasaki}(2024)}]{tasaki2024macroscopic}%
  \BibitemOpen
  \bibfield  {author} {\bibinfo {author} {\bibfnamefont {H.}~\bibnamefont {Tasaki}},\ }\bibfield  {title} {\bibinfo {title} {Macroscopic irreversibility in quantum systems: Eth and equilibration in a free fermion chain},\ }\href@noop {} {\bibfield  {journal} {\bibinfo  {journal} {arXiv preprint arXiv:2401.15263}\ } (\bibinfo {year} {2024})}\BibitemShut {NoStop}%
\bibitem [{\citenamefont {Haferkamp}\ \emph {et~al.}(2021)\citenamefont {Haferkamp}, \citenamefont {Bertoni}, \citenamefont {Roth},\ and\ \citenamefont {Eisert}}]{haferkamp2021emergent}%
  \BibitemOpen
  \bibfield  {author} {\bibinfo {author} {\bibfnamefont {J.}~\bibnamefont {Haferkamp}}, \bibinfo {author} {\bibfnamefont {C.}~\bibnamefont {Bertoni}}, \bibinfo {author} {\bibfnamefont {I.}~\bibnamefont {Roth}},\ and\ \bibinfo {author} {\bibfnamefont {J.}~\bibnamefont {Eisert}},\ }\bibfield  {title} {\bibinfo {title} {Emergent statistical mechanics from properties of disordered random matrix product states},\ }\bibfield  {journal} {\bibinfo  {journal} {PRX Quantum}\ }\textbf {\bibinfo {volume} {2}},\ \href {https://doi.org/10.1103/prxquantum.2.040308} {10.1103/prxquantum.2.040308} (\bibinfo {year} {2021})\BibitemShut {NoStop}%
\bibitem [{\citenamefont {Csisz{\'a}r}\ and\ \citenamefont {Talata}(2006)}]{csiszar2006context}%
  \BibitemOpen
  \bibfield  {author} {\bibinfo {author} {\bibfnamefont {I.}~\bibnamefont {Csisz{\'a}r}}\ and\ \bibinfo {author} {\bibfnamefont {Z.}~\bibnamefont {Talata}},\ }\bibfield  {title} {\bibinfo {title} {Context tree estimation for not necessarily finite memory processes, via bic and mdl},\ }\href {https://doi.org/10.1109/TIT.2005.864431} {\bibfield  {journal} {\bibinfo  {journal} {IEEE Transactions on Information theory}\ }\textbf {\bibinfo {volume} {52}},\ \bibinfo {pages} {1007} (\bibinfo {year} {2006})}\BibitemShut {NoStop}%
\bibitem [{\citenamefont {Milman}\ and\ \citenamefont {Schechtman}(1986)}]{milman1986asymptotic}%
  \BibitemOpen
  \bibfield  {author} {\bibinfo {author} {\bibfnamefont {V.~D.}\ \bibnamefont {Milman}}\ and\ \bibinfo {author} {\bibfnamefont {G.}~\bibnamefont {Schechtman}},\ }\href@noop {} {\emph {\bibinfo {title} {Asymptotic theory of finite dimensional normed spaces: Isoperimetric inequalities in riemannian manifolds}}},\ Vol.\ \bibinfo {volume} {1200}\ (\bibinfo  {publisher} {Springer Science \& Business Media},\ \bibinfo {year} {1986})\BibitemShut {NoStop}%
\end{thebibliography}%

\clearpage
\onecolumngrid

\setcounter{equation}{0}
\renewcommand{\theequation}{S.\arabic{equation}}

\appendix

\section{Characterization of the macroscopic states (Proof of Theorem \ref{th:macro-uniform})}
\label{appendix:characterization}

\begin{lemma}\label{lemma:groupings}
    Suppose that $\povm{Q}=\{Q_y\}_{y\in\set{Y}}$ is a PVM satisfying $Q_y Q_{y'} = \delta_{yy'} Q_y$ for all $y, y' \in \set{Y}$. Suppose also that there exists another POVM $\povm{P}=\{P_x\}_{x\in\set{X}}$ such that $\povm{Q}\preceq \povm{P}$. Then, the post-processing transforming $\povm{P}$ into $\povm{Q}$ is deterministic, i.e., $p(y|x)\in\{0,1\}$ for all $x$ and $y$.
\end{lemma}

\begin{proof}
    Fix an arbitrary element $y \in \mathcal{Y}$ and let $|\phi_{y}\rangle$ be a unit vector satisfying the condition $\langle \phi_{y} | Q_{y} | \phi_{y} \rangle = 1$. Then, we have
    \begin{align}
        1&=\braket{\phi_{y}|Q_{y}|\phi_{y}}\\
        &=\sum_{x:p(y|x)>0}p(y|x)\braket{\phi_{y}|P_x|\phi_{y}}\\
        &\le\sum_{x:p(y|x)>0}\braket{\phi_{y}|P_x|\phi_{y}}\\
        &\le 1\;.
    \end{align}
    Without loss of generality, we assume $\braket{\phi_{y}|P_x|\phi_{y}} \ne 0$ for $x\in\set{X}$. When $\braket{\phi_{y}|P_x|\phi_{y}}=0$, take $|\phi_{\tilde{y}}\rangle$ such that $\langle \phi_{\tilde{y}} | Q_{\tilde{y}} | \phi_{\tilde{y}} \rangle = 1$ and $\braket{\phi_{\tilde{y}}|P_x|\phi_{\tilde{y}}}\neq0$. The existence of such $|\phi_{\tilde{y}}\rangle$ follows from the fact that the eigenvectors with eigenvalue 1 of the elements of $\{Q_y\}_{y\in\set{Y}}$ form an orthonormal basis. Therefore, from $p(y|x)\braket{\phi_{y}|P_x|\phi_{y}}=\braket{\phi_{y}|P_x|\phi_{y}}$, we obtain $p(y|x)=1$. Since $\sum_{y}p(y|x)=1$, it follows that $p(y'|x)=0$ for any $y'\neq y$. Thus, 
    \begin{align}
        p(y|x)\in\{0,1\}\;
    \end{align}
    for all $y\in\set{Y}$ and $x\in\set{X}$.
\end{proof}

\begin{proof}[Proof of Theorem \ref{th:macro-uniform}]
    By definition, $\macro$ is a macroscopic state if and only if
    \begin{align}\label{eq:S5}
        D\left(\macro\middle\|u\right)=D\left(\mP(\macro)\middle\|\mP(u)\right).
    \end{align}
    As proved in~\cite{buscemi2022observational}, the above condition is equivalent to
    \begin{align}
        \macro=\sum_x\tr{P_x\;\macro}\frac{P_x}{V_x}\;.
        \label{eq:conditionmmacro}
    \end{align}
    This conclusion is obtained simply by noticing that the set of macroscopic states coincides with the set of fixed point of the composite channel $[\mcl{R}_{\mP, u} \circ \mP](\cdot)$, where $\mcl{R}_{\mP, u}$ is the Petz transpose map~\cite{petz1986sufficient,petz1988sufficiency} of $\mP(\cdot):=\sum_x\tr{P_x\;\cdot}\ketbra{x}$, computed with respect to the maximally mixed prior $u$, that is
    \begin{align}
        \mcl{R}_{\mP, u}(\cdot):=\sum_x\bra{x}\cdot\ket{x}\;\frac{P_x}{\tr{P_x}}\;.
    \end{align}
    It is easy to verify that
    \begin{align}
        [\mcl{R}_{\mP, u} \circ \mP](\cdot)=\sum_x\tr{P_x\;\cdot}\frac{P_x}{\tr{P_x}}\;.
    \end{align}
    According to the terminology introduced in~\cite{buscemi2022observational}, $[\mcl{R}_{\mP, u} \circ \mP](\rho)$ is the coarse-grained version of $\rho$, done with respect to the measurement $\povm{P}$. A state is macroscopic for $\povm{P}$ if and only if it coincides with its own coarse-grained version.
    
    
    The condition (\ref{eq:conditionmmacro}) is satisfied for $\macro':=\sum_yc_y\Pi_y$. Indeed, we have
    \begin{align}
        \sum_x\tr{P_x\;\macro'}\frac{P_x}{V_x}&=\sum_y c_y\sum_x\tr{P_x\Pi_y}\frac{P_x}{V_x}\\
        &=\sum_y c_y\sum_{x:{\rm supp}{P_x}\subset{\rm supp}\Pi_y}\tr{P_x}\frac{P_x}{V_x}\\
        &=\sum_yc_y\Pi_y\\
        &=\macro'\;,
    \end{align}
    where the second and the third line follows from Lemma \ref{lemma:groupings}.
    
    Conversely, suppose that $\macro$ satisfies Eq.~\eqref{eq:S5} and let $\{|e_k\rangle\}_k$ be an eigenbasis of $\macro$. Correspondingly, let us define the map
    \begin{align*}
    \diag_{\macro}(X):=\sum_k\ketbra{e_k}X\ketbra{e_k}\;.    
    \end{align*}
    As shown in~\cite[Section 4.2]{petz2003monotonicity}, Eq.~\eqref{eq:S5} implies that
    \begin{align}\label{eq:to-be-summed}
        \macro\Pi_y=c_y\Pi_y\;,
    \end{align}
    where $\Pi_y:=\sum_{x\in\set{X}_y}\diag_{\macro}(P_x)$, for some partition of $\set{X}$ into disjoint subsets as $\set{X}=\set{X}_1\cup\set{X}_2\cup\cdots\cup\set{X}_\ell$, such that $\diag_{\macro}(P_x)\diag_{\macro}(P_{x'})=0$ whenever $x$ and $x'$ do not belong to the same subset. Notice that, in general, such a partition is not uniquely defined.
    
    It is easy to see that the operators $\Pi_y$ are, by construction, positive semi-definite and orthogonal to each other. Moreover, since $\openone=\diag_{\macro}(\openone)=\diag_{\macro}(\sum_xP_x)=\sum_x\diag_{\macro}(P_x)$, the set  $\povm{\Pi}=\{\Pi_y\}_y$ in fact constitutes a PVM. Thus, by summing~\eqref{eq:to-be-summed} over $y$, we immediately obtain Eq.~\eqref{eq:macro-explicit} in the main text, i.e.,
    \begin{align*}
        \macro=\sum_yc_y\Pi_y\;.
    \end{align*}
    
    We still need to show that, for any such a partition, the corresponding PVM $\povm{\Pi}$ satisfies $\povm{\Pi}\preceq\povm{P}$. This is due to the fact that, since $\ker\diag_{\macro}(P_x)\subseteq\ker P_x$, i.e., $\supp\diag_{\macro}(P_x)\supseteq\supp P_x$, two POVM elements $P_x$ and $P_{x'}$ must have orthogonal supports, whenever $x$ and $x'$ belong to different subsets. Thus, not only $\Pi_y=\sum_{x\in\set{X}_y}\diag_{\macro}(P_x)$, but in fact $\Pi_y=\sum_{x\in\set{X}_y}P_x$, that is $\povm{\Pi}\preceq\povm{P}$.

\end{proof}

In particular, if $\macro$ is macroscopic for $\povm{P}=\{P_x\}$, then
\begin{align*}
    [\macro,P_x]=0\;,
\end{align*}
for all $P_x\in\povm{P}$. In order to show this, for any PVM $\povm{\Pi}=\{\Pi_y\}$ such that $\povm{\Pi}\preceq\povm{P}$, as a consequence of Lemma~\ref{lemma:groupings} we have $\Pi_y:=\sum_{x\in\set{X}_y}P_x$, where $\set{X}_y$ are disjoint subsets covering $\set{X}$. Then, $\supp\Pi_y\supseteq\supp P_x$ for all $x\in\set{X}_y$, and being $\Pi_y$ a projection, we see that $\Pi_y$ acts as the identity operator on the supports of all the $P_x$'s it comprises. Moreover, it acts as the null operator on the support of all remaining $P_x$'s. Thus, any $\Pi_y$ commutes with any $P_x$, though we notice that the operators $P_x$ need not commute with each other. As a consequence, since macroscopic states are just linear combinations of $\Pi_y$'s, they all commute with all $P_x$'s.

\section{Concentration of observational entropy}
\label{appendix:concentration}

The proof strategy that we use here is to show that, for a unitary sampled at random, the probability that $\tr{U\rho U^\dag\; P}$ is ``close'' to $\tr{u\;P}$, for any density matrix $\rho$ and any non-null effect $0\le P \le\openone$, is ``high''. We then use this fact to bound the probability of large deviations in observational entropy.

\subsection{From probability to observational entropy}

\begin{lemma}\label{lemma:concg}
Let $\rho$ be a density operator and $\povm{P}=\{P_x\}_{x\in\mathcal{X}}$ a POVM such that $|\mathcal{X}|<+\infty$. 
Let $\nu$ denote a measure on the $d$-dimensional unitary group $\mathcal{U}_d$ and $\mathbb{P}_\nu$ the corresponding probability of an event.
Suppose that there exists a real-valued non-increasing function $g$ such that
\begin{align}
    \mathbb{P}_\nu\left\{\left|{\rm Tr}[U\rho U^\dagger P_x]-{\rm Tr}[u P_x]\right|\geq\xi\right\}\leq g(\xi)\;,
\end{align}
for all $x\in\mathcal{X}$ and $\xi>0$. Then, it holds that
\begin{align}
    \mathbb{P}_\nu\left\{S_{\povm{P}}(U\rho U^\dagger)\leq(1-\delta)\log{d}\right\}\leq \frac{1}{\kappa(\povm{P})}\:g\left(\kappa(\povm{P})\sqrt{\delta\log{d}}\right),\label{eq:nuNg}
\end{align}
where $\kappa(\povm{P}):=\min_{x\in\mathcal{X}}{\rm Tr}[u P_x]$.

\begin{proof}
    Define probability distributions $\mathbf{p}=\{p_x\}_x$ and $\mathbf{q}=\{q_x\}_x$ by $p_x={\rm Tr}[U\rho U^\dagger P_x]$ and $q_x={\rm Tr}[u P_x]$. By the definition of the observational entropy, we have
    \begin{align}
        \log{d}-S_{\povm{P}}(U\rho U^\dagger)=D_{\rm KL}(\mathbf{p}\|\mathbf{q})\;,
    \end{align}
    where $D_{\rm KL}$ is the Kullback-Leibler divergence for classical probability distributions defined by $D_{\rm KL}(\mathbf{p}\|\mathbf{q})=\sum_xp_x\log{(p_x/q_x)}$.
    The condition $S_{\povm{P}}(U\rho U^\dagger)\leq(1-\delta)\log{d}$ is thus equivalent to $D_{\rm KL}(\mathbf{p}\|\mathbf{q})\geq\delta\log{d}$. 
    We will invoke the following upper bound on $D_{\rm KL}$: 
    \begin{align}\label{eq:csiszar}
        D_{\rm KL}(\mathbf{p}\|\mathbf{q})\leq\sum_x\frac{(p_x-q_x)^2}{q_x}\leq\frac{|\mathcal{X}|\max_x(p_x-q_x)^2}{\min_xq_x}\;,
    \end{align}
    where the first inequality was proved in~\cite{csiszar2006context} (see Lemma A.3 therein).
    The condition $D_{\rm KL}(\mathbf{p}\|\mathbf{q})\geq\delta\log{d}$ then implies
    \begin{align}
        \frac{|\mathcal{X}|\max_x(p_x-q_x)^2}{\kappa(\povm{P})}\geq\delta\log{d},
    \end{align}
    which is equivalent to
    \begin{align}
        &\max_x\left|{\rm Tr}[U\rho U^\dagger P_x]-{\rm Tr}[u P_x]\right|\geq\sqrt{\frac{\kappa(\povm{P})\delta\log{d}}{|\mathcal{X}|}}.
    \end{align}
    Thus, we have
    \begin{align}
        \mathbb{P}_\nu\left\{S_{\povm{P}}(U\rho U^\dagger)\leq(1-\delta)\log{d}\right\}
        &\leq\mathbb{P}_\nu\left\{\max_x\left|{\rm Tr}[U\rho U^\dagger P_x]-{\rm Tr}[u P_x]\right|\geq\sqrt{\frac{\kappa(\povm{P})\delta\log{d}}{|\mathcal{X}|}}\right\}\\
        &\leq\sum_x\mathbb{P}_\nu\left\{\left|{\rm Tr}[U\rho U^\dagger P_x]-{\rm Tr}[u P_x]\right|\geq\sqrt{\frac{\kappa(\povm{P})\delta\log{d}}{|\mathcal{X}|}}\right\}\\
        &\leq |\mathcal{X}|\:g\left(\sqrt{\frac{\kappa(\povm{P})\delta\log{d}}{|\mathcal{X}|}}\right)\\
        &\leq \frac{1}{\kappa(\povm{P})}\:g\left(\kappa(\povm{P})\sqrt{\delta\log{d}}\right),
    \end{align}
    where the last line is due to $|\mathcal{X}|\kappa(\povm{P})\leq1$.
\end{proof}

\end{lemma}

\subsection{Concentration with Haar-random unitaries (Proof of Theorem \ref{th:concentration})}

\begin{lemma}\label{lemma:unitarylevy}
(Lemma 3.2 in \cite{low2009large}, see also \cite{ledoux-2001,milman1986asymptotic})
Let $f$ be a Lipschitz function on $\mathcal{U}_d$, with the Lipschitz constant $\eta$ defined by
\begin{align}
    \eta=\sup_{U_1\neq U_2\in\mathcal{U}_d}\frac{|f(U_1)-f(U_2)|}{\|U_1-U_2\|_2}\;.
\end{align}
Then
\begin{align}
    \mathbb{P}_H(|f-\mathbb{E}_H[f]|\geq\xi)\leq4\exp\left(-\frac{2d}{9\pi^3\eta^2}\xi^2\right),
\end{align}
where $\mathbb{P}_H$ denotes the probability computed with respect to the Haar (unitarily invariant) measure on $\mathcal{U}_d$, and $\mathbb{E}_H[f]$ denotes the corresponding mean of $f$.
\end{lemma}

\begin{lemma}\label{lemma:LipschitzUP}
Let $\rho$ be a density operator and let $P$ be a positive semidefinite operator such that $P\leq \openone$.
For a Haar-distributed random unitary, it holds that
    \begin{align}
        \mathbb{P}_H\left\{\left|\tr{U\rho U^\dagger P}-\tr{uP}\right|\geq\xi\right\}\leq4\exp\left(-\frac{d\xi^2}{18\pi^3}\right)\;.\label{eq:HaarconcUP}
    \end{align}
\end{lemma}

\begin{proof}
Due to Lemma \ref{lemma:unitarylevy}, it suffices to prove that the Lipschitz constant of the function $f:U\mapsto{\rm Tr}[U\rho U^\dagger P]$ is bounded above by $2$.
By the triangle inequality and the Cauchy–Schwarz inequality, we have
\begin{align}
    \left|{\rm Tr}[U_1\rho U_1^\dagger P]-{\rm Tr}[U_2\rho U_2^\dagger P]\right|
    &=\left|\frac{1}{2}{\rm Tr}[(U_1+U_2)\rho (U_1-U_2)^\dagger P]+\frac{1}{2}{\rm Tr}[(U_1-U_2)\rho (U_1+U_2)^\dagger P]\right|\\
    &\leq\frac{1}{2}\left|{\rm Tr}[(U_1+U_2)\rho (U_1-U_2)^\dagger P]\right|+\frac{1}{2}\left|{\rm Tr}[(U_1-U_2)\rho (U_1+U_2)^\dagger P]\right|\\
    &\leq \left\|U_1-U_2\|_2\cdot\|P(U_1+U_2)\rho\right\|_2.
\end{align}
Taking a diagonal decomposition $\rho=\sum_i\lambda_i\ketbra{e_i}$, the second term in the last line can be bounded as
\begin{align}
    \left\|P(U_1+U_2)\rho\right\|_2^2&=\tr{\rho(U_1+U_2)^\dagger P^2(U_1+U_2)\rho}\\
    &\leq\tr{\rho(U_1+U_2)^\dagger (U_1+U_2)\rho}\\
    &=\sum_{i=1}^d \lambda_i^2\langle e_i|(U_1+U_2)^\dagger (U_1+U_2)|e_i\rangle\\
    &\leq\sum_{i=1}^d \lambda_i^2\left\|(U_1+U_2)\right\|_\infty^2\\
    &\leq4\tr{\rho^2}\;.
\end{align}
This implies
\begin{align}
    \frac{\left|{\rm Tr}[U_1\rho U_1^\dagger P]-{\rm Tr}[U_2\rho U_2^\dagger P]\right|}{\left\|U_1-U_2\right\|_2}\leq2\sqrt{\tr{\rho^2}}\leq2\;,
\end{align}
and completes the proof.
\end{proof}

\begin{proof}[Proof of Theorem \ref{th:concentration}]
It follows from Lemma \ref{lemma:LipschitzUP} that
\begin{align}
        \mathbb{P}_H\left\{\left|\tr{U\rho U^\dagger P_x}-\tr{uP_x}\right|\geq\xi\right\}\leq4\exp\left(-\frac{d\xi^2}{18\pi^3}\right)
    \end{align}
for all $x$.
Applying Lemma \ref{lemma:concg}, we obtain 
\begin{align}
    \mathbb{P}_H\left\{S_{\povm{P}}(U\rho U^\dagger)\leq(1-\delta)\log{d}\right\}&\leq  \frac{4}{\kappa(\povm{P})}\exp\left(-\frac{\delta}{18\pi^3}\kappa(\povm{P})^{2}d\log{d}\right)\;,
\end{align}
where $\kappa(\povm{P}):=\min_{x\in\mathcal{X}}{\rm Tr}[u P_x]$.
\end{proof}

\subsection{Concentration with approximate unitary designs (Proof of Theorem \ref{th:con1})}

Next, we consider the concentration of observational entropy by approximate designs.
Various definitions of approximate unitary designs have been proposed, depending on what measure is used to define ``approximate''.
We adopt the following version proposed in \cite{brandao2021models}:
\begin{definition}(Definition 4 in \cite{brandao2021models})\label{def:diamond2}
    Fix $\ep>0$ and $t\in\mathbb{N}$. A unitary ensemble $\mathcal{E}:=\{p_i, U_i\}_{i=1}^{N}$ is an $\ep$-approximate (unitary) $t$-design in the diamond distance if 
    \begin{align}
        \norm{M^{(t)}_{\mathcal{E}}-M^{(t)}_{H}}_{\diamond}\leq\frac{t!}{d^{2t}}\ep\;.
    \end{align}
    Here, $M^{(t)}_{\mathcal{E}}(X):=\sum_{i=1}^{n}p_i U_{i}^{\ot t}X(U_{i}^{\dag})^{\ot t}$ and $M^{(t)}_{H}(X):=\mathbb{E}[U^{\ot t}X(U^{\dag})^{\ot t}]$.
\end{definition}

\begin{lemma}\label{lemma:UPconcHt2}
    Let $\rho$ be a density operator and let $P$ be a positive semidefinite operator such that $P\leq \openone$. It holds that
    \begin{align}
        \mathbb{P}_{{\mE}}\left\{\left|\tr{U\rho U^\dagger P}-\tr{uP}\right|\geq\xi\right\}\leq\frac{1}{\xi^{t}}(1+\ep)\left(\frac{t^{2}}{d}\right)^{\frac{t}{2}}.
    \end{align}
    Here, $U$ is sampled from an $\ep$-approximate $t$-design in the diamond distance.
\end{lemma}

\begin{proof}
    By Markov's inequality, we have
     \begin{align}
        \mathbb{P}_{{\mE}}\left\{\left|\tr{U\rho U^\dagger P}-\tr{uP}\right|\geq\xi\right\}&=\mathbb{P}_{{\mE}}\left\{\left|\tr{U\rho U^\dagger P}-\tr{uP}\right|^t\geq\xi^t\right\}\\
        &\leq\frac1{\xi^{t}}\mathbb{E}_{\mE}\left[\left|\tr{U\rho U^\dagger P}-\tr{uP}\right|^{t}\right].
    \end{align}
    Due to the convexity of $f(z)=|z|^{t}$, with $\sum_{i=1}^{r}\lam_{i}\ketbra{i}$ be the spectral decomposition of $\rho$, we have
    \begin{align}
        \mathbb{E}_{\mE}\left[\left|\tr{U\rho U^\dagger P}-\tr{u\;P}\right|^{t}\right]
        &=\mathbb{E}_{\mE}\left[\left|\sum_{i=1}^{r}\lam_{i}(\tr{U\ketbra{i} U^\dagger P}-\tr{uP})\right|^{t}\right]\\      &\leq\sum_{i=1}^{r}\lam_{i}\mathbb{E}_{\mE}\left[\left|\tr{U\ketbra{i} U^\dagger P}-\tr{uP}\right|^{t}\right].
        \label{eq:lamEr}
    \end{align}
    We will invoke Corollary 24 in \cite{brandao2021models}:
    Given a pure state $\ket{\phi}$ and a positive semidefinite operator $P \leq \openone$, it holds 
    \begin{align}
        \mathbb{E}_\mE\left[\left|\tr{U\ketbra{\phi} U^\dagger P}-\tr{uP}\right|^t\right]\leq(1+\ep)\left(\frac{t^{2}}{d}\right)^{\frac{t}{2}}\;,
    \end{align}
    where $U$ is sampled from an $\ep$-approximate (in the diamond measure) $t$-design $\mE$.
    Applying this to each term in (\ref{eq:lamEr}), we complete the proof.
\end{proof}

\begin{proof}[Proof of Theorem \ref{th:con1}]
    It follows from Lemma \ref{lemma:concg} and Lemma \ref{lemma:UPconcHt2} that
    \begin{align}
        \mathbb{P}_{\mE}\left\{S_{\povm{P}}(U\rho U^\dagger)\leq(1-\delta)\log{d}\right\}
        &\leq\frac{1+\ep}{\kappa(\povm{P})}\left(\frac{t^{2}}{\kappa(\povm{P})^{2}\delta d\log{d}}\right)^{\frac{t}{2}}\;,\label{inq:2t-design2}
    \end{align}
    with $U$ being randomly sampled from an $\ep$-approximate  $t$-design $\mE$ in the diamond distance.  Substituting $t=2$ yields Ineq.~(\ref{eq:OE-concentration-design}).
\end{proof}

\begin{remark}
It was proved in \cite{harrow2023approximate} that approximate unitary designs can be implemented on $D$-dimensional lattices by a local random quantum circuit of depth polynomial in the lattice size.  
Though the measures adopted in \cite{harrow2023approximate} to define ``$\ep$-approximate'' are different from Definition \ref{def:diamond2} above, it affects the circuit depth only polynomially.
Hence, we conclude that the concentration of observational entropy occurs even under random unitaries generated by random polynomial-depth quantum circuits, which is often regarded as a physically more realistic model of local quantum chaotic dynamics. In other words, our setup is a toy model of the physical ETH many-body Hamiltonian~\cite{kaneko2020characterizing}.
\end{remark}

\end{document}